\documentclass{article}
\usepackage{fullpage}
\usepackage{graphicx, amsthm}
\usepackage{balance}  
\usepackage{cite, subfig, times}
\usepackage{bm}
\usepackage{graphicx,color}
\usepackage{xspace, amsfonts,float,verbatim,enumitem,caption,xparse,algorithm,algorithmic}
\usepackage{hyperref}

\usepackage{epsfig,bm,epsf,float, amsopn, amsfonts, varwidth, cite, url, amsmath}
\usepackage{xspace}
\usepackage{etoolbox}

\newcommand{\eps}{\varepsilon}
\newcommand{\MG}{\textbf{MG}}
\newcommand{\SSH}{\textbf{SSH}}
\newcommand{\RBMC}{\textbf{RBMC}}
\newcommand{\MED}{\textbf{MED}}
\newcommand{\SMED}{\textbf{SMED}}
\newcommand{\SMIN}{\textbf{SMIN}}
\newcommand{\RTUC}{\textbf{RTUC}}
\renewcommand{\SS}{\textbf{SS}}
\newcommand{\MHE}{\textbf{MHE}}

\newtheorem{theorem}{Theorem}

\newtheorem{lemma}{Lemma}

\begin{document}


\title{A High-Performance Algorithm for Identifying Frequent Items in Data Streams}

\pagestyle{plain}
\date{}
\author{
Daniel Anderson\thanks{Georgetown University.}
\and
Pryce Bevin\thanks{Georgetown University.}
\and
Kevin Lang\thanks{Yahoo.}
\and
Edo Liberty\thanks{Amazon. Research performed while at Yahoo Research.}
\and
Lee Rhodes\thanks{Yahoo.}
\and
Justin Thaler\thanks{Georgetown University. Parts of this work were performed while the author was at Yahoo Research.}
}

\maketitle

\begin{abstract}
Estimating frequencies of items over data streams is a common building block in streaming data measurement and analysis. 
Misra and Gries introduced their seminal algorithm for the problem in 1982, and the problem has since been revisited many times due its practicality and applicability. 
We describe a highly optimized version of Misra and Gries' algorithm that is suitable for deployment in industrial settings.
Our code is made public via an open source library called DataSketches that is already used by several companies and production systems. 

Our algorithm improves on two theoretical and practical aspects of prior work.
First, it handles {\it weighted} updates in amortized constant time, a common requirement in practice.
Second, it uses a simple and fast method for merging summaries that asymptotically improves on prior work even for unweighted streams.
We describe experiments confirming that our algorithms are more efficient than prior proposals. 

\end{abstract}

\section{Introduction}
Identifying frequent items (also known as heavy hitters) and answering point queries 
(i.e., queries of the form ``approximately how many times did item $i$ appear in the stream?) are two of the most basic computational tasks 
performed on data streams. 
Due to their practical importance, streaming algorithms for these tasks
have been studied intensely \cite{mg, spacesaving, lossycounting, countmin, countsketch, beatingcount, berinde, ch, woodruffpods, merge1, merge2}.  
It may therefore seem as though streaming frequency approximation is well-understood,
with little room for further insight or improvement. However, when we set about implementing an algorithm
suitable for industrial use on web-scale data, we found that existing algorithms
had significant shortcomings. 
In this work, we describe these shortcomings and present an optimized algorithm that addresses them. 

\subsection{Our Contributions} 
Our algorithm is similar to the popular Misra-Gries (\MG) \cite{mg} and Space Saving (\SS) \cite{spacesaving} algorithms.
However, it differs from them in several simple yet important ways. 
Specifically, our algorithm handles weighted updates in a highly efficient manner. 
It runs in amortized constant time per stream update, and employs optimized data structures to minimize its memory footprint while maximizing throughput. 
We also give a simple and efficient method for merging summaries produced by our algorithm, or more generally, any \emph{counter-based} algorithm (cf. Section \ref{sec:counter} for a definition of this class of algorithms). Our merging procedure significantly improves on prior work
even for unweighted streams.
We present experimental results validating the superiority 
of our algorithms over existing solutions. 

Our experiments 
also clarify conventional wisdom regarding the best way to implement 
frequent items algorithms in practice. 
Prior literature strongly suggested that the best algorithm in practice
is a min-heap based implementation
of \SS\ (we denote this implementation of \SS\ by \MHE, cf. Section \ref{sec:mhe}). 
Indeed, experimental work of Cormode and Hadjieleftheriou \cite{ch} 
that focused exclusively on the setting of unweighted streams identified \SS\ 
as state of the art.
They showed that a linked-list based implementation of \SS\ (denoted in \cite{ch} as \textbf{SSL}) is noticeably
faster than the min-heap based implementation, but also significantly more space intensive.
Moreover, \textbf{SSL} does not naturally extend to weighted updates, while \MHE\ does.
Based on these findings, subsequent work (e.g., \cite{hhh}) in the weighted setting used \MHE\ as the implementation
of choice.
We overturn this conventional wisdom: our implementations do not
use a min-heap, yet are
significantly faster and more space efficient than
\MHE.

\subsection{Problem Statement and Motivation}
Given a data stream $\sigma$ of length $n$ over a universe $[m]=\{1, \dots, m\}$, where each stream update
is of the form $(i_j, \Delta_j)$ for some $i_j \in [m]$ and $\Delta_j > 0$ is a weight, let $f_i = \sum_{j \colon i_j = i} \Delta_j$ denote the (weighted)
frequency of $i$ in the stream and let $N = \sum_j \Delta_j$ denote the weighted stream length. 
The goal of an algorithm for answering point queries is the following. 
In a single pass over the stream, the algorithm must compute a small summary (sometimes also called a \emph{sketch})\footnote{Some works use the term sketch to refer only to a specific class of summaries known as \emph{linear sketches}.} of the stream
such that, for any item $i \in [m]$, it is possible to efficiently derive an estimate $\hat{f}_i$
such that $|\hat{f}_i-f_i|$ is small. 

Ideally, the algorithm should also be able to identify $(\phi, \eps)$-heavy hitters (in the $\ell_1$-norm). That is, given a user-specified threshold $\phi$,
the algorithm should be able to identify all items $i$ with $f_i \geq \phi \cdot N$, possibly also returning a small number
of ``false positive'' items whose frequency is slightly below the specified threshold $\phi N$.

We note that algorithms for answering point queries and identifying heavy hitters are often used in a black-box fashion as subroutines in
streaming algorithms solving more complicated problems. These include estimating the \emph{empirical entropy} of a data stream \cite{entropy},
and identifying \emph{hierarchical heavy hitters} \cite{hhh}, both of which have important applications
in network data analysis and anomaly detection \cite{entropyapp1, entropyapp2, entropyapp3, hhhapp1, hhhapp2}. 
Our optimized algorithms for answering point queries and identifying heavy hitters
can be directly substituted into these more complicated algorithms. While we leave this task
to future work, we expect that our optimizations would yield
similar speed and memory improvements in these settings.

\paragraph*{The Importance of Weighted Updates}
\label{sec:apps}
\label{apps}
The majority of the literature on streaming approximate frequency analysis 
considers the case of unweighted updates, where $\Delta_j=1$ for all $j$. 
However, there are many applications in which the ability to efficiently handle weighted updates is critical.
For example:
\begin{itemize}
\item It is often useful to track  the total amount of data users send over a network, or the total time that users spend using a mobile app or watching online video.
\item In spam detection, it is useful to track the total number of recipients of emails from a given address, and many emails
have an enormous number of recipients.
\item Weighted updates arise naturally whenever items from the data universe are associated with importance factors. For example, 
in information-retrieval and text mining, term frequency--inverse document frequency (tf-idf for short) is a standard statistic
that assigns a value to each word in a document that is part of a larger corpus. The value assigned to a word increases with the number of times a word appears in the document, but decreases with the frequency of the word in the corpus, in order to account for the fact that some words appear more frequently in general.
\item We will also see (cf. Section \ref{sec:merging}) that the ability to efficiently handle weighted updates enables
extremely efficient \emph{merging} of summaries -- significantly faster than merging procedures proposed in prior work.
\end{itemize}
Unlike algorithms tailored to unweighted updates, algorithms capable of handling weighted updates typically apply to real-valued weights. This will be the case for the algorithms we give in this work.

There are a number of simple ways to modify algorithms for the unweighted setting to handle the weighted setting. However, 
as we describe in Section \ref{sec:priorworkweighted}, all of the modifications that (to our knowledge) have been proposed in prior work
have significant shortcomings.

\subsection{Overview of Prior Work}
\label{sec:counter}
In a comprehensive 
paper and survey, Cormode and Hadjieleftheriou \cite{ch} provide a detailed comparison of 
proposed algorithms for finding heavy hitters and answering point queries in the case of unit weight updates. 
They classified algorithms into three classes: counter-based algorithms,
quantile algorithms, and (linear) sketches. They found that counter-based algorithms perform significantly better in terms of space, speed, and accuracy than quantile and sketching
algorithms, a finding that we confirmed in our own initial experiments.
Hence, we focus on counter-based algorithms for the remainder
of this work. 

\paragraph*{Note.} The main advantage of sketches over counter-based algorithms is the ability to handle streams with deletions (i.e., streams in the strict turnstile model, in which the $\Delta_j$'s may be negative, yet for each $i$ the final frequency
$f_i$ of $i$ is non-negative). 
Specifically, in the strict turnstile model, sketch algorithms
can return frequency estimates with error proportional to $N=\sum_{j} \Delta_j$.
While deletions are not encountered in the applications we target in this work (cf. Section \ref{sec:apps}),
we note here that any counter-based algorithm can easily be used to handle deletions at the cost of 
having error proportional to $\sum_{j} |\Delta_j|$ rather than to $N=\sum_{j} \Delta_j$.
To achieve this, simply apply one instance of a counter-based algorithm to all the positive-weight updates,
and a second instance to (the absolute value of) of all negative weight updates, and define the estimate for $f_i$ to be the 
difference of the estimates for $i$ returned by the two instances. By the triangle inequality, the error of any estimated frequency
is at most the sum of the errors of the two sketches. This approach will be suitable
in applications in which $\sum_{j} |\Delta_j|$ is not much larger than $N=\sum_j \Delta_j$.

\subsubsection{Counter-Based Algorithms} 
A counter-based algorithm stores $k$ counters, where $k$ is a parameter of the algorithm that controls both the space usage
and the error (larger $k$ corresponds to more space and less error). Each counter stores an approximate count for
some item of the data universe; the various counter-based algorithms differ in how they assign counters to items and how they determine
the approximate counts. 

Cormode and Hadjieleftheriou \cite{ch} found that two counter-based algorithms are state of the art: Misra-Gries (MG) \cite{mg} and Space Saving (SS) \cite{spacesaving}. 
They found that the estimates returned by \SS\ tend to be more accurate than those returned by \MG\ -- we discuss the
 reasons for the improved accuracy of \SS\ relative to \MG\ in Section \ref{sec:estimate}.
However, subsequent work by Agarwal et al. \cite{mergeable} observed that these two algorithms are actually isomorphic,
in the following sense. The estimates returned by the \SS\ algorithm with $k+1$ counters can be derived from
the summary computed by the \MG\ algorithm with $k$ counters. That is, the summaries computed 
by both \SS\ and \MG\ contain essentially the same information about the stream, but the algorithms
differ in how they use this information to determine a ``best estimate'' for any point query.

Our algorithm can be naturally viewed as an extension and modification of \MG. Hence, we now describe \MG\ in detail.

\subsubsection{Description of the Misra-Gries Algorithm For Unit Weight Updates} 
Every time the algorithm processes a stream update $(i, +1)$, the \MG\ algorithm looks to see if $i$ is assigned a counter,
and if so it increments the counter. If not, and an unassigned counter exists, the algorithm assigns the counter to $i$ and sets
the approximate count to $1$. If no unassigned counter exists, the algorithm decrements all counters by 1, and marks all counters
set to 0 as unassigned. When asked to provide an estimate $\hat{f}_i$ for the frequency of item $f_i$,
the algorithm returns 0 if $i$ is not assigned a counter, and returns the value of the counter assigned to $i$ otherwise.

See Algorithm \ref{code:mg} for pseudocode. The following standard lemma bounds the error of any estimate returned by the \MG\ algorithm \cite{mg}, and we provide the proof for completeness.

{\small
\begin{algorithm}[t]
{\small
\caption{\small{Misra-Gries Algorithm for Unit Weight Updates}}\label{code:mg}
\begin{algorithmic}[1]{
\STATE \textbf{Algorithm}: Misra-Gries(k):}
\STATE \hspace{2mm} $T \gets \emptyset$ // $T$ is set of items assigned a counter
\STATE \textbf{Function } Update($i, +1$):
\STATE \hspace{2mm} \textbf{if } $i \in T$:
\STATE \hspace{4mm} $c(i) \gets c(i) + 1$
\STATE \hspace{2mm} \textbf{else if } $|T| < k$:
\STATE \hspace{4mm} $T = T \cup \{i\}$
\STATE \hspace{4mm} $c(i) \gets 1$
\STATE \hspace{2mm} \textbf{else}:
\STATE \hspace{4mm} DecrementCounters()
\STATE \textbf{Function } DecrementCounters():
\STATE \hspace{2mm} \textbf{for all} $j \in T$:
\STATE \hspace{4mm} $c(j) = c(j) - 1$
\STATE \hspace{4mm} \textbf{if } $c(j) = 0$:
\STATE \hspace{6mm} $T = T \setminus \{j\}$
\STATE \textbf{Function } Estimate($i$):
\STATE \hspace{2mm} \textbf{if } $i \in T$ 
\STATE \hspace{4mm} \textbf{return } $c(i)$
\STATE \hspace{2mm} \textbf{else }
\STATE \hspace{4mm} \textbf{return } 0

\end{algorithmic}
}
\end{algorithm}
}

\begin{lemma} The \MG\ algorithm with $k$ counters is guaranteed to return, for each $i \in [n]$, and estimate $\hat{f}_i$ satisfying
$0 \leq f_i - \hat{f}_i \leq N/(k+1)$. \label{lem:die} \end{lemma}
\begin{proof}
It is obvious from the description of the algorithm that $0 \leq f_i - \hat{f}_i$ for all $i$.
To prove that $f_i - \hat{f}_i \leq N/k$ for all $i$, observe that $f_i - \hat{f}_i$ is at most the total
number of decrement operations over the course of the algorithm, and there can be at most $N/k$ decrement operations.
Indeed, if there are $d$ decrement operations,
then since each decrement operation affects all $k$ counters, the sum of the counter values at the end of
the algorithm is exactly $N-d \cdot (k+1)$. Since no counter value is ever negative, it holds that $N-d\cdot (k+1) \geq 0$, and hence $d \leq N/(k+1)$.
\end{proof}

Using a similar proof, Berinde et al. \cite{berinde} show that the \MG\ algorithm also satisfies the following guarantee.
This guarantee, referred to by Berinde et al. as a \emph{tail guarantee}, is much stronger than Lemma \ref{lem:die} for streams with a very skewed frequency distribution (i.e., with a relatively small number of frequent
 items
constituting the bulk of the stream).

\begin{lemma}[Berinde et al. \cite{berinde}] \label{lem:berinde}
Let $N^{\text{res}(j)}$ denote
the sum of the frequencies of all but the top $j$ most frequent items. 
For any $j < k$, the \MG\ algorithm with $k$ counters is guaranteed to return, for each $i \in [n]$, and estimate $\hat{f}_i$ satisfying
$0 \leq f_i - \hat{f}_i \leq N^{\text{res}(j)}/(k+1-j)$.
\end{lemma}


\noindent \textbf{Efficiently Implementing Misra-Gries for Unit Weight Updates.}
A natural way to implement the \MG\ algorithm is to maintain a hash table that stores all assigned counters. The key for a counter is the item it is assigned to, and the value is the count. Whenever a stream update $(i, +1)$ arrives, the algorithm looks up key $i$ in the hash table, and if it is present, it increments the corresponding value (i.e., count). If not, the algorithm checks to see if the hash table is storing $k$ key-value pairs. If not, the key-value pair $(i, 1)$ is inserted into the hash table. If so, the algorithm iterates over all key-value pairs in the hash table, decrementing each value and deleting any pair if the value becomes zero. Assuming all hash table operations (i.e., lookup, insert, and delete) take $O(1)$ amortized time, and that enumerating all key-value pairs in the hash table can be done in time $O(k)$, then this implementation runs in amortized $O(1)$ time per stream update. This amortized time bound exploits the fact that decrement operations, which require enumerating all $k$ key-value pairs, only occur at most $N/k$ times. 

\subsubsection{Description of the Space Saving Algorithm For Unit Weight Updates}
\label{sec:ssunit}
When processing a stream update $(i, +1)$, the \SS\ algorithm behaves identically to \MG, except in the case
where $i$ is not assigned a counter, and no unassigned counter exists. In this case, \SS\ increments the counter with
the smallest value and assigns it to $i$. 
When asked to provide an estimate $\hat{f}_i$ for item the frequency of item $f_i$,
the algorithm returns the value of the smallest counter if $i$ is not assigned a counter, and returns the value of the counter assigned to $i$ otherwise.
See Algorithm \ref{code:ss} for pseudocode.

{\small
\begin{algorithm}[t]
{\small
\caption{\small{Space Saving Algorithm for Unit Weight Updates}}\label{code:ss}
\begin{algorithmic}[1]{
\STATE \textbf{Algorithm}: SpaceSaving(k):}
\STATE \hspace{2mm} $T \gets \emptyset$ // $T$ is set of items assigned a counter
\STATE \textbf{Function } Update($i, +1$):
\STATE \hspace{2mm} \textbf{if } $i \in T$:
\STATE \hspace{4mm} $c(i) \gets c(i) + 1$
\STATE \hspace{2mm} \textbf{else if } $|T| < k$:
\STATE \hspace{4mm} $T = T \cup \{i\}$
\STATE \hspace{4mm} $c(i) \gets c(i) + 1$
\STATE \hspace{2mm} \textbf{else }:
\STATE \hspace{4mm} $j \gets \text{arg min}_{j \in T} c(j)$
\STATE \hspace{4mm} $c(i) \gets c(j) + 1$
\STATE \hspace{4mm} $T = \left(T \setminus \{j\}\right) \cup \{i\}$
\STATE \textbf{Function } Estimate($i$):
\STATE \hspace{2mm} \textbf{if } $i \in T$:
\STATE \hspace{4mm} \textbf{return } $c(i)$
\STATE \hspace{2mm} \textbf{else }
\STATE \hspace{4mm} \textbf{return } \text{min}$_{j \in T}\{c(j)\}$

\end{algorithmic}
}
\end{algorithm}
}


\paragraph*{Efficiently Implementing Space Saving for Unit Weight Updates}
The \SS\ algorithm can be naturally implemented by using a min-heap data structure to track the smallest counter at all times. 
Following \cite{ch}, we refer to this implementation as \SSH. 
Unfortunately, the use of a min-heap is slow (taking time $O(\log k)$ per stream update), and it also nearly doubles the space usage compared to the \MG\ 
implementation described above,
since both a hash table of capacity $k$ and a min-heap must be stored.

In the case of unit updates, Metwally et al. \cite{spacesaving} propose a doubly linked list based data structure they call \emph{Stream Summary} that can process 
updates in $O(1)$ time. However, the need to store pointers in the linked list will more than double the space usage of the Misra-Gries implementation described above. Moreover, as mentioned again in Section \ref{sec:diediedie} below, this implementation method does
not naturally extend to weighted updates.

\subsubsection{Extending Misra-Gries to the Weighted Case: Prior Work and Its Limitations}
\label{sec:priorworkweighted}
There are a number of simple ways to extend the \MG\ algorithm to handle weighted updates.
Perhaps the simplest is to treat an update $(i_j, \Delta_j)$ as $\Delta_j$ unit updates. 
We refer to this as the \emph{Reduce-To-Unit-Case} (\RTUC-\MG) extension of the \MG\ algorithm.
However, this takes time at least $\Delta_j$ per stream update, which
is unacceptable when the weights may be large (this approach is also awkward to apply to the case of non-integer weights). 

To our knowledge, the only other proposal to extend the \MG\ algorithm to weighted updates is due to Berinde et al. \cite{berinde}, who suggest the following modification.
Every time the algorithm processes a stream update $(i_j, \Delta_j)$, the algorithm looks to see if $i_j$ is assigned a counter,
and if so it increments the counter by $\Delta_j$. If not, and an unassigned counter exists, the algorithm assigns the counter to $i$ and sets
the approximate count to $\Delta_j$. If no unassigned counter exists, the algorithm's behavior depends on
whether $\Delta_j \leq c_{\text{min}}$, 
where $c_{\text{min}}$ denotes the minimum value of any counter. 
If $\Delta_j \leq c_{\text{min}}$, then all stored counters are reduced by $\Delta_j$. Otherwise,
all counters are reduced by $c_{\text{min}}$, and some counter with zero
count (there must be at least one now) is assigned to $i_j$ and given
count $\Delta_j-c_{\text{min}}$.
We refer to this as the \emph{Reduce-By-Min-Counter} (\RBMC) extension of the \MG\ algorithm to weighted updates. 

It is easy to see that the \RBMC\ algorithm produces estimates identical to the \RTUC-\MG\ algorithm, and hence
satisfies the accuracy guarantees of Lemmas \ref{lem:die} and \ref{lem:berinde}. The 
advantage of the \RBMC-\MG\ algorithm is in its runtime, which does not grow linearly with $\Delta_j$.
The main shortcoming of the \RBMC\ algorithm is that it still may not run in amortized $O(1)$ time per stream update. In
fact, there are streams on which the proposal will perform expensive decrement operations on essentially every stream update: consider a stream where the first $k$ updates increment 
the counts of distinct items by an arbitrarily large number $M$, and then the following $M$ stream updates are unit updates to different items. Decrement operations, each requiring time $\Theta(k)$, will be performed on every one of the last $M$ updates in the stream. While such examples may be contrived, Section \ref{sec:expts} provides experimental results on real datasets showing that the runtime of this algorithm is significantly higher in practice than the alternatives presented in this paper.

\subsubsection{Extending Space Saving to the Weighted Case: Prior Work and Its Limitations}
\label{sec:diediedie} \label{sec:mhe}
The \SS\ algorithm also has a natural \emph{Reduce-To-Unit-Case} ({\RTUC-\SS}) extension to weighted updates.
However, like the \RTUC-\MG\ alogrithm, this takes time at least $\Delta_j$ per stream update, which
is unacceptable when the weights may be large. 

The min-heap based implementation, \SSH, of \SS\ naturally extends to the weighted case. We refer to this as the \emph{Min-Heap-Extension} (\MHE) extension of the \SS\ algorithm to weighted updates. However, \MHE\ suffers from
the same poor (i.e., $O(\log k)$) update time and larger space usage (relative to \RBMC) as in the unit weight update case.
Despite these shortcomings, this implementation of \SS\ for weighted updates was used in at least one prior work on computing Hierarchical Heavy Hitters \cite{hhh}.

The implementation of Space Saving that uses the \emph{Stream Summary} data structure does not naturally extend to the case of weighted updates \cite{ch}.


\medskip
We discuss additional prior work in Section \ref{sec:additional}.

\subsection{Isomorphism Results for Weighted Streams} 
\label{sec:isomorphic}
We mention that the \MHE\ and \RBMC\ algorithms are isomorphic, in the sense that 
the estimates returned by the \MHE\ algorithm with $k+1$ counters can be derived from
the summary computed by the \RBMC\ algorithm with $k$ counters. This follows immediately from the facts
that both the \MHE\ and \RBMC\ produce the same estimates as \RTUC-\SS\ and \RTUC-\MG\ respectively,
and the latter two algorithms were shown to be isomorphic by Agarwal et al. \cite{mergeable}. 

\section{Our Algorithm}
For expository purposes, we first propose an initial modification of the \MG\ algorithm that achieves (up to a constant factor)
the same accuracy guarantees as the \RBMC\ algorithm, while guaranteeing $O(1)$ amortized runtime per stream update.
We then describe our final algorithm, which improves over the runtime and space usage of the initial modification by a factor close to 2.

\subsection{An Initial Proposal}
For simplicity, we assume throughout this section that $k$ is even.
Recall (cf. Section \ref{sec:priorworkweighted}) that the main downside of the \RBMC\ algorithm \cite{berinde} is that it does not guarantee $O(1)$ amortized runtime per stream update.
The reason was that \RBMC\ may have to perform decrement operations essentially every stream update,
and each decrement operation requires iterating over $k$ counters.

To guarantee $O(1)$ amortized time per stream update, it is enough to ensure that decrement operations are performed
at most once every (say) $k/2$ stream updates. We mention that a similar technique was used by Liberty to efficiently identify ``frequent directions'' in a stream of vectors \cite{edo}.

 To ensure that this property holds, it suffices to consider the following 
simple modification of \RBMC. 
Every time the algorithm processes a stream update $(i_j, \Delta_j)$, the algorithm looks to see if $i_j$ is assigned a counter,
and if so it increments the counter by $\Delta_j$. If not, and an unassigned counter exists, the algorithm assigns the counter to $i_j$ and sets
the approximate count to $\Delta_j$. If no unassigned counter exists, the algorithm decrements all counters by 
the \emph{median} counter value $c_{\text{median}}$. The algorithm then marks all counters
set to a non-positive value as unassigned. If $\Delta_j$ is larger than $c_{\text{median}}$,
then $i_j$ is then assigned a counter, which is set to $\Delta_j - c_{\text{median}}$.
 
We refer to this algorithm as the \emph{Reduce-By-Median-Counter} (\MED) extension of the \MG\ algorithm to weighted updates.
See Algorithm \ref{code:firstweighted} for pseudocode, which is expressed in slightly more general form.
Specifically, in Algorithm \ref{code:firstweighted},
$c_{\text{median}}$ is replaced by the the $k^*$'th largest counter value,
where $k^*$ is a parameter of the algorithm. One recovers the description above by setting $k^*=k/2$.

{\small
\begin{algorithm}[t]
{\small
\caption{\small{The Reduce-By-Median-Counter (\MED) Extension of the \MG\ Algorithm to Weighted Updates}}\label{code:firstweighted}
\begin{algorithmic}[1]{
\STATE // Notation: $k^*$ is a parameter of the algorithm
\STATE \textbf{Algorithm}: Initial-Algorithm(k):}
\STATE \hspace{2mm} 
$T \gets \emptyset$ // $T$ is set of items assigned a counter
\STATE \textbf{Function } Update($i, \Delta$):
\STATE \hspace{2mm} \textbf{if } $i \in T$:
\STATE \hspace{4mm} $c(i) \gets c(i) + \Delta$
\STATE \hspace{2mm} \textbf{else if } $|T| < k$:
\STATE \hspace{4mm} $T = T \cup \{i\}$
\STATE \hspace{4mm} $c(i) \gets c(i) + \Delta$
\STATE \hspace{2mm} \textbf{else }:
\STATE \hspace{4mm} DecrementCounters()
\STATE \hspace{4mm} \label{linedie} \textbf{if } $\Delta \geq c_{k^*}$ // See Line \ref{notation} for definition of $ c_{k^*}$
\STATE \hspace{6mm} \label{linediep1}$T = T \cup \{i\}$ // $|T| \leq k^*+1 \leq k$ after this line
\STATE \hspace{6mm} \label{linedie2} $c(i) \gets c(i) + \Delta - c_{k^*}$
\STATE \textbf{Function } DecrementCounters(): \label{diediedie}
\STATE \label{notation} \hspace{2mm} // Notation: Let $c_{k^*}$ be the $k^*$-largest value, counting multiplicity,\\  \hspace{2mm} // in the multiset $\{c(j) \colon j \in T\}$.
\STATE \hspace{2mm} \textbf{for all} $j \in T$:
\STATE \hspace{4mm} $c(j) = c(j) - c_{k^*}$ \label{theline}
\STATE \hspace{4mm} \textbf{if } $c(j) \leq 0$:
\STATE \hspace{6mm} $T = T \setminus \{j\}$ \label{die2} \label{line:unassign}
\STATE \textbf{Function } Estimate($i$):
\STATE \hspace{2mm} \textbf{if } $i \in T$:
\STATE \hspace{4mm} \textbf{return } $c(i)$
\STATE \hspace{2mm} \textbf{else }
\STATE \hspace{4mm} \textbf{return } 0

\end{algorithmic}
}
\end{algorithm}
}
\subsubsection{Runtime and Accuracy Analysis}
\begin{lemma}  \label{lemma:runtime}
In Algorithm \ref{code:firstweighted}, a DecrementCounters() operation (Lines \ref{diediedie}-\ref{line:unassign}) is performed at most once every $k^*$ stream updates.
\end{lemma}
\begin{proof}
Observe that DecrementCounters() operations are performed only when all $k$ counters are assigned. 
Every time a DecrementCounters()  operation is executed, Line \ref{theline} sets at least $k^*$ counters to a non-positive value, and all such
counters become unassigned in Line \ref{line:unassign}.
Since any Update$(i, \Delta)$ operation assigns at most one counter to $i$, it follows that after a DecrementCounters() operation, at least $k^*$ 
Update() operations must occur before all counters once again become assigned.
\end{proof}

Lemma \ref{lemma:runtime} is easily seen to imply that Algorithm \ref{code:firstweighted} can be implemented in $O(1)$ amortized time per stream update.
The main observation is that $c_{k^*}$ (cf. Line \ref{notation}) can be found in time $O(k)$ using the Quickselect algorithm \cite{quickselect}. Hence,
decrement operations can be performed in $O(k)$ time, and by Lemma \ref{lemma:runtime}, decrement operations occur at most once every $k^*=\Omega(k)$ stream updates. 

\begin{theorem} \label{thm:runtime}
By maintaining all $k$ counters in a hash table of capacity supporting (amortized) constant time insertions, deletions, and lookups,
and enumeration of all counters in $O(k)$ time,
Algorithm \ref{code:firstweighted} can be implemented in amortized constant time as long as the parameter $k^* \geq \Omega(k)$.
\end{theorem}

It is also not difficult to show that this implementation satisfies the same error guarantee as Lemma \ref{lem:die} up to a factor of 2.
In fact, we prove the following \emph{tail guarantee} that is significantly stronger than Lemma \ref{lem:die} for streams with skewed frequency distributions. Our analysis exploits the key ideas of earlier analyses (of the \MG\ algorithm) by 
Agarwal et al. \cite{mergeable} and Berinde et al. \cite{berinde}.

\begin{theorem} \label{thm:errorweighted} For any $j < k^*$, Algorithm \ref{code:firstweighted} is guaranteed to return, for each $i \in [n]$, an estimate $\hat{f}_i$ satisfying
$$0 \leq f_i - \hat{f}_i \leq N^{\text{res}(j)}/(k^*-j).$$
\end{theorem}
\begin{proof}
It is obvious from the description of the algorithm that $0 \leq f_i - \hat{f}_i$ for all $i$.
To prove the second inequality, we will proceed by induction on the number
of stream updates $n$. For clarity, we introduce some notation. 
Let $f_{i, n} = \sum_{j \leq n: i_j = i} \Delta_j$ denote the (weighted) frequency of item $i$ 
in the first $n$ stream updates, and similarly let $\hat{f}_{i, n}$ denote the output of Estimate($i$)
after the first $n$ stream updates.
Let $E_n = \max_i f_{i, n} - \hat{f}_{i, n}$ denote the maximum error of any estimate that might be returned by the summary
after processing $n$ stream updates. 
Let $C_n=\sum_{i \in T} c(i)$ denote the sum of all counts in the summary after the first $n$ stream updates have been processed. Let $N_n = \sum_{i=1}^n \Delta_i$
denote the weighted stream length after $n$ stream updates.

For expository purposes, we establish the second inequality for $j=0$,
before establishing it for any $j < k^*$.

\medskip\noindent \textbf{Establishing the second inequality for $j=0$.}
In the case $j=0$
it is sufficient to show establish the following lemma.
\begin{lemma} \label{key} $ E_n \leq (N_n-C_n)/k^*$. \end{lemma}
\begin{proof}
We show this by induction on $n$. Clearly it holds for $n=0$.
Assume by way of induction that it holds for streams consisting of $n-1$ updates.

\medskip
\noindent \textit{Proof Outline for Lemma \ref{key}.}
The key observations are (a) that error is only introduced by calls to DecrementCounters()
and (b) whenever DecrementCounters() is called,
at least $k^*$ counters are reduced by $c_{k^*}$. Hence, 
the error in any estimate increases by $c_{k^*}$, and the difference between
$N$ and $C$ increases by at least $k^* \cdot c_{k^*}$. This ensures that the
difference between $N$ and $C$ is always at least $k^*$ times the error
of the summary, as desired. The rest of the proof
makes this argument formal.

\medskip

Suppose first that the $n$th stream update $(i_n, \Delta_n)$ does not cause DecrementCounters() to be called.
In this case, after processing the $n$th stream update, $f_{i_n}$ and $\hat{f}_{i_n}$
both increase by $\Delta_n$, and for any $i \neq i_n$, $f_i$ and $\hat{f}_i$ are unchanged. Hence,
$E_n=E_{n-1}$, i.e., the error of the summary is unchanged by the $n$th stream update. 
Meanwhile, $N_n$ and $C_n$ both increase by $\Delta_n$.
It follows that $(N_n-C_n) = (N_{n-1}-C_{n-1})$. Putting these facts together and invoking the inductive hypothesis,
we conclude as desired that
 \begin{equation*} E_n = E_{n-1} \leq (N_{n-1}-C_{n-1})/k^* = (N_{n}-C_{n})/{k^*} .\end{equation*}
 
 Next, suppose that the $n$th stream update does cause DecrementCounters() to be called. 
 We assume w.l.o.g. that $\Delta_n \leq c_{k^*}$; if not, then 
 we instead treat update $(i_n, \Delta_n)$ as two separate 
 updates $(i_n, c_{k^*})$ followed by $(i_n, \Delta_n-c_{k^*})$, as the algorithm
 behaves identically in this case.
 
Then the largest $k^*$ counters are each decremented by $c_{k^*}$, 
 and no counter is incremented (as Lines \ref{linediep1}-\ref{linedie2} are not executed due to the assumption).
 Hence, 
 \begin{align} \notag N_n - C_n \!=\!  & N_{n-1} \!+\! \Delta_n \!-\! C_{n} \geq N_{n-1}\! + \! \Delta_n \!-\! (C_{n-1}\! -\! k^* \cdot c_{k^*}) \\
 \notag   = & N_{n-1}  - C_{n-1} +  k^* \cdot c_{k^*} + \Delta_n\\
  \geq & N_{n-1}  - C_{n-1} +  k^* \!\cdot\! c_{k^*}. \label{frigg}\end{align}
 Meanwhile, it holds that \begin{equation} \label{frigg2} E_n \leq E_{n-1} + c_{k^*}.\end{equation} 
To see this, first observe that, for all $i \neq i_n$, 
 $\hat{f}_{i, n} - \hat{f}_{i, n-1} \leq c_{k^*}$, and $f_{i, n} = f_{i, n-1}$.
 Hence, $f_{i, n} - \hat{f}_{i, n} \leq  f_{i, n-1} - \hat{f}_{i, n-1} + c_{k^*}$, i.e.,
 the error in the estimate for $i$ increased by at most $c_{k^*}$ when processing the $n$th stream update.
 
 To see that Equation \eqref{frigg2} also holds for $i_n$ itself, observe first that
 the fact that DecrementCounters()
 was called when processing the $n$th stream update means that $i_n$ was not assigned a counter after $n-1$
 stream updates. Hence, $\hat{f}_{i_n,n-1}=0$. Meanwhile, the assumption that $\Delta_n \leq c_{k^*}$ then ensures
 that $\hat{f}_{i_n, n} = 0$ as well. And clearly
 $f_{i_n, n} = f_{i_n, n-1} + \Delta_n$. Putting these facts together implies that
 $$f_{i_n, n} - \hat{f}_{i_n, n} \!=\! f_{i_n, n-1} \! + \! \Delta_n \!-\! \hat{f}_{i_n, n-1} \leq f_{i_n, n-1} + c_{k^*} \!-\! \hat{f}_{i_n, n-1}.$$
That is, the error in the estimate for $i_n$ increased by at most $c_{k^*}$ when processing the $n$th stream update.
Equation \eqref{frigg2} follows.

Combining Equations \eqref{frigg} and \eqref{frigg2} with the inductive hypothesis, we conclude as desired that
\begin{align*} E_n & \leq E_{n-1} + c_{k^*} \leq (N_{n-1} - C_{n-1})/k^* + c_{k^*}\\
& \leq (N_n - C_n - k^* c_{k^*})/k^* + c_{k^*} = (N_n - C_n)/k^*.
 \end{align*}
 \end{proof}
 \medskip
\noindent \textbf{Establishing the second inequality for any $j < k^*$.}
For notational simplicity, let us assume that $f_1 \geq f_2 \geq \dots \geq f_j$. 
Lemma \ref{key} implies that
\begin{equation}
\label{ihatethis} C_n \leq N_n - k^* \cdot E_n.
\end{equation}
Meanwhile, it is clear that at the end of the stream, $c(i) \geq f_i - E_n$ for all $i$,
which implies that \begin{equation}
\label{seriouslyl} C_n \geq \sum_{i=1}^{j} (f_i - E_n) \geq \left(\sum_{i=1}^{j} f_i\right) - j \cdot E_n. \end{equation}
Putting Equations \eqref{ihatethis} and \eqref{seriouslyl} together, we conclude that 
$$\left(\sum_{i=1}^{j} f_i\right) -  j \cdot E_n \leq N_n - k^* \cdot E_n,$$
which in turn implies that
$$(k^*-j) \cdot E_n \leq N_n  - \sum_{i=1}^j f_i = N^{\text{res}(j)}.$$
Hence, $E_n \leq N^{\text{res}(j)}/(k^*-j).$
This completes the proof.
\end{proof}

\subsection{The Final Algorithm}
Algorithm \ref{code:firstweighted} has two disadvantages that are important in practice.
First, an extra $k$ words of space are required during every DecrementCounters() operation,
on top of the $\Theta(k)$ words of space required for that hash table that maintains the counters.
This extra space is required in order to find the $k^*$'th largest counter, as the Quickselect algorithm (or any sorting procedure)
cannot be done in place without destroying the hash table. This disadvantage nearly
doubles the space usage of the algorithm.

Second, during each DecrementCounters() operation,
the algorithm must make an extra pass through the summary to find the $k^*$'th largest counter, before decrementing all counters and discarding the
non-positive ones. As DecrementCounters() operations is the time bottleneck in practice,
this significantly slows the concrete efficiency of the algorithm.

In order to address both of these disadvantages at once, we make the following observation.
In neither the proof of Theorem \ref{thm:runtime} (establishing constant amortized update time for Algorithm \ref{code:firstweighted}) nor Theorem \ref{thm:errorweighted} (establishing strong bounds on the error of Algorithm \ref{code:firstweighted}) is it crucial that we decrement by the $k^*$'th largest counter value for $k^*=k/2$. 
The property that ensures constant amortized update time as per Theorem \ref{thm:runtime} is that each DecrementCounters() operation
decrements by a value that is \emph{larger than} a constant fraction of the counters.
The property exploited to establish the error bounds of Theorem \ref{thm:errorweighted} is that each DecrementCounters() operations decrements by a value that is \emph{smaller than} a constant fraction of the counters.

Both properties can be ensured by modifying the DecrementCounters() function as follows. We randomly sample $\ell$ counters in the hash table, where $\ell=O(\log n)$ is a suitably chosen parameter. We then compute the median of the sampled counters, and decrement all counters by this value, discarding any counters that are non-positive after decrementing. We refer to this algorithm as the \emph{Reduce-By-Sample-Median} (\SMED) extension of the \MG\ algorithm. See Algorithm \ref{code:final} for pseudocode.

\begin{theorem} \label{thm:runtimefinal}
There is an $\ell = O(\log n)$ such that Algorithm \ref{code:firstweighted} can be implemented to run in amortized constant time with probability at least $1-1/n$.
\end{theorem}
\begin{proof}
Standard Chernoff bounds (e.g., \cite[Theorem 4.4]{mu}) imply that if $\ell \geq c \cdot \log n$ for a suitably large constant $c > 0$,
then with probability at least $1-1/n^2$ the median $c^*$ of $\ell$ counters sampled from $\{c(i) \colon i \in T\}$ 
will satisfy $c^* \geq c(i)$ for at least $k/3$ values of $i$ (where recall that $k=|T|$ whenever a DecrementCounters() operation occurs).
By a union bound, we conclude that this event occurs for \emph{all} DecrementCounters() operations
with probability at least $1-1/n$. In this event, DecrementCounters() operations occur
at most once every $k/3$ stream updates. The claimed time bound now follows
from the same argument as in Theorem \ref{thm:runtime}.
\end{proof}

\begin{theorem} \label{thm:errorfinal}
For any constant $c > 2$, there is an $\ell = O(\log n)$ such that the following holds with probability at least $1-1/n$.
For any $j < k/c$, Algorithm \ref{code:final} will return, for each $i \in [n]$, and estimate $\hat{f}_i$ satisfying
$$0 \leq f_i - \hat{f}_i \leq N^{\text{res}(j)}/(k/c-j).$$

\end{theorem}
\begin{proof}
It is easy to see that the proof of Theorem \ref{thm:errorweighted} applies even if Algorithm \ref{code:final}
decrements all counters by \emph{at most} $c_{k^*}$ rather than exactly $c_{k^*}$.
Hence, it suffices to show that 
for any constant $c > 2$, there is an $\ell = O(\log n)$ such that with probability at least $1-1/n^2$,
the median $c^*$ of $\ell$  
counters sampled from $\{c(i) \colon i \in T\}$ satisfies 
$c^* \leq c(i)$ for at least $k/c$ values of $i \in T$.
Indeed, this property holds by standard Chernoff bounds  \cite[Theorem 4.4]{mu}. By combining this property with
a union bound over all DecrementCounters() operations, we conclude that with probability at least $1-1/n$,
Algorithm \ref{code:final} always decrements all counters by at most $c_{k^*}$ where $k^*=k/c$. 
\end{proof}

{\small
\begin{algorithm}[t]
{\small
\caption{\small{The Reduce-By-Sample-Median (\SMED) Extension of the \MG\ Algorithm to Weighted Updates}}\label{code:final}
\begin{algorithmic}[1]{
\STATE // Notation: $k^*$ is a parameter of the algorithm
\STATE \textbf{Algorithm}: Initial-Algorithm(k):}
\STATE \hspace{2mm} $T \gets \emptyset$ // $T$ is set of items assigned a counter
\STATE \textbf{Function } Update($i, \Delta$):
\STATE \hspace{2mm} \textbf{if } $i \in T$:
\STATE \hspace{4mm} $c(i) \gets c(i) + \Delta$
\STATE \hspace{2mm} \textbf{else if } $|T| < k$:
\STATE \hspace{4mm} $T = T \cup \{i\}$
\STATE \hspace{4mm} $c(i) \gets c(i) + \Delta$
\STATE \hspace{2mm} \textbf{else }:
\STATE \hspace{4mm} DecrementCounters()
\STATE \hspace{4mm} \label{linediefinal} \textbf{if } $\Delta \geq c^*$ // See Line \ref{finalline} for definition of $ c^*$
\STATE \hspace{6mm} \label{linediep1final}$T = T \cup \{i\}$ // $|T| \leq k^*+1 < k$ after this line
\STATE \hspace{6mm} \label{linedie2final} $c(i) \gets c(i) + \Delta - c^*$
\STATE \textbf{Function } DecrementCounters(): \label{diediediefinal}
\STATE \label{notationfinal} \hspace{2mm} // Notation: $\ell$ is a parameter
\STATE \hspace{2mm} \label{finalline}// Randomly sample $\ell$ counters from $T$\\
 \hspace{2mm} // and let $c^*$ be the median of the samples.
\STATE \hspace{2mm} \textbf{for all} $j \in T$:
\STATE \hspace{4mm} $c(j) = c(j) - c^*$ \label{thelinefinal}
\STATE \hspace{4mm} \textbf{if } $c(j) \leq 0$:
\STATE \hspace{6mm} $T = T \setminus \{j\}$ \label{die2final} \label{line:unassignfinal}
\STATE \textbf{Function } Estimate($i$):
\STATE \hspace{2mm} \textbf{if } $i \in T$:
\STATE \hspace{4mm} \textbf{return } $c(i)$
\STATE \hspace{2mm} \textbf{else }
\STATE \hspace{4mm} \textbf{return } 0

\end{algorithmic}
}
\end{algorithm}
}

\subsection{Additional Implementation Details}
We have completed a highly optimized implementation of Algorithm \ref{code:final} in Java
and incorporated it into an open source library \cite{datasketches} that is 
widely deployed within Yahoo and several other companies. We briefly present additional details of our implementation not covered above.

\subsubsection{Variant of Function Estimate()} \label{sec:estimate}
In our implementation, we modified the function Estimate() compared to the pseudocode 
of Algorithm \ref{code:final}. Specifically, our implementation keeps a variable $\text{offset}$ that tracks the sum of all 
decrement values $c^*$ over all DecrementCounters() operations. If $i$ is assigned a counter,
then our implementation's estimate $\hat{f}_i$ for $f_i$ is $c(i) + \text{offset}$. 
If $i$ is not assigned a counter, the returned estimate is 0. 

This represents a hybrid of the estimates returned by the \MG\ and \SS\ algorithms (cf. Algorithms \ref{code:mg} and \ref{code:ss} respectively). It behaves analogously to the \MG\ algorithm on items $i$ not assigned counters -- outputting 0 in this case --
and (as explained below) analogously to the \SS\ algorithm on items $i$ that are assigned counters. 

Our reasoning for this choice is as follows. Before the \MG\ and \SS\ algorithms were realized to be isomorphic \cite{mergeable}, Cormode and Hadjieleftheriou \cite{ch} found that the estimates returned by \SS\ were superior in practice, despite the fact that the worst-case error bounds for both algorithms are essentially identical. The reason for this
is that the \SS\ estimates are ``more aggressive'' -- in particular, whereas the \MG\ algorithm will always strictly underestimate the 
frequency of an item (unless fewer than $k$ items appear in the stream), the \SS\ algorithm may well output \emph{exactly correct} estimates of frequent items.
The estimates returned by our actual implementation inherit this desirable property of \SS.

On the other hand, the \SS\ algorithm has the undesirable property that it will always strictly \emph{overestimate} 
the frequency of an item that does not appear in the stream (unless fewer than $k$ items appear in the stream),
whereas the \MG\ algorithm always outputs exactly correct estimates (namely, 0) for such items. 
The estimates returned by our actual implementation inherit this desirable property of \MG.

We remark that, like \MG\ and \SS\ themselves, our implementation is also capable of outputting tight upper and lower bounds on the frequency of $i$.
The upper bound is $c(i) + \text{offset}$, while the lower bound is $c(i)$ (or 0 if $i$ is not assigned a counter).

\subsubsection{Choice of $\ell$}
We performed numerical calculations establishing that the choice $\ell=1024$ guarantees 
that the following holds. 
For streams of length $N \leq 10^{20}$,
Algorithm \ref{code:final}, with probability at least $1 - 1.5 \times 10^{-8}$, 
returns estimates $\hat{f}_i$ for all $i \in [m]$ satisfying the following error guarantee:
$$0 \leq f_i - \hat{f}_i \leq N^{\text{res}(j)}/(.33 \cdot k-j).$$
Hence, we set $\ell=1024$ in our implementation.

\subsubsection{Hash Table and Exact Space Usage} 
We experimented with a wide variety of hash table implementations for storing counters. While many 
performed similarly, 
we discovered that a good choice is to use linear probing. Specifically, keys and values are kept
in two parallel arrays of length $L \approx 4k/3$ (the length is rounded up to the closest power of two
to ensure fast modular arithmetic on indices into the arrays, in case a probe ``falls off the end'' of the array).
We keep an additional array 
of state variables, where a state of 0 indicates that a cell is unoccupied,
and a positive state equals the distance (plus one) of the stored key $x$ from its preferred cell $h(x)$,
where $h \colon [m] \to L$ is the hash function.
Insertion and lookup operations are processed as in standard linear probing: the implementation
successively looks at cells $h(x), h(x)+1 (\text{mod } L), \dots, $ until the key is found or an empty cell is discovered.

The other functionality that must be supported by our hash table is decrementing all values by a specified 
amount $c^*$ and deleting all that become non-positive. This is done by starting at the end (i.e., largest index) 
of a run (i.e., a contiguous sequence) of occupied cells; each time a non-positive value is identified, the implementation deletes it,
and then works its way toward the end of the run, shifting keys and values forward as necessary to ensure
that all future lookup and insert operations behave correctly. 

State variables need only consist of 2 bytes with overwhelming probability (we performed numerical calculations showing that when $k \leq 2^{32}$ and $L = 4k/3$, the probability that at any given time a state variable exceeds $2^{14}$ is at most $10^{-250}$).
Hence, assuming item identifiers and approximate counts are 8 bytes each, and that $4k/3$
is a power of 2, our implementation uses $18 \cdot (4/3) \cdot k = 24 \cdot k$ bytes, plus a small constant number of additional bytes.

\section{Merging Summaries}
\label{sec:merge}
An extremely important functionality in real systems is the ability to efficiently
merge the summaries of separate data sets (via an arbitrary \emph{aggregation tree})
to obtain a summary of the union of the data sets,
without increasing the size of the summary or its approximation
error. This mergeability property enables a wide variety 
of scenarios. For example, consider a system that keeps a separate summary for many subsets of a large dataset.
This might occur if, e.g., a company keeps a separate summary for data obtained in each 1-hour period over the course of several years.
The company may further subdivide the data by, say, geographic region or other attributes. At query time, an analyst can specify which data are of interest 
(e.g., the analyst may be interested in users from a given state who used the company's service in a given 48-hour window). The summaries can then
be seamlessly merged to answer approximate queries about the data of interest. In this important example, merging is occurring at query time, and millions or billions of summaries may need to be merged to answer the query. Hence, merging must be extremely efficient, justifying our focus in this section on developing a highly optimized merging procedure. 

As a second example, mergeable summaries enable a single large dataset to be partitioned arbitrarily among many machines,
with each machine processing each partition separately. The resulting 
summaries can then be merged in any order to obtain a summary for the complete dataset.
Hence, mergeable summaries enable datasets to be automatically processed in a parallel and distributed manner.

As a final example, consider multiple datasets that are geographically distributed. It may be infeasible to transmit the entirety of any
one of the datasets to a central machine. If an analyst is interested in the union of the datasets, mergeable summaries
enable each dataset to be
summarized separately, and merely the summaries transmitted to one machine.

\subsection{Earlier Merging Procedures} 
\paragraph*{The Procedure of Berinde et al. \cite{berinde}}
The merging procedure we describe in this section is highly similar to a proposal of Berinde et al. \cite{berinde},
who described their proposal in the context of arbitrary counter-based summaries. Berinde et al. suggested that, in order
to merge many summaries, one should treat each counter in each summary as a stream update (so a counter $c(i)$
assigned to item $i$ is treated as a stream update $(i, c(i))$) and feed all the stream updates into a single new summary.
They proved that if the counter-based algorithm when run on a single stream can return estimated frequencies $\hat{f}_i$ for each $i$
satisfying $|f_i - \hat{f}_i| \leq N^{\text{res}(j)}/(k-j),$
then the output of the merging procedure can return estimated frequencies $\hat{f}_i$ satisfying 
\begin{equation} \label{merge} |f_i - \hat{f}_i| \leq 3 \cdot N^{\text{res}(j)}/(k-2j), \end{equation}
where the quantities $f_i$ and $N^{\text{res}(j)}$ in Equation \eqref{merge} refer to the \emph{concatenation}
of the streams that were fed into each of summaries being merged.

The proposal and accuracy analysis of Berinde et al. \cite{berinde} have two significant limitations. 
First, the accuracy analysis only guarantees that the merged summary has at most \emph{triple} the error of the constituent summaries.
Second, as pointed out in \cite{mergeable}, when merging many summaries, the proposal does not support an 
arbitrary aggregation tree (i.e., it is not possible to repeatedly merge pairs of summaries in an arbitrary manner until a single summary is obtained). Rather, each summary must be merged ``into'' a single output summary. This second limitation is really a different manifestation of the first: the reason that Berinde et al. cannot support an arbitrary aggregation tree is that their error analysis cannot rule out the possibility that each edge in the tree causes a tripling of the error. Hence, the final error bound for an arbitrary aggregation tree would be \emph{exponentially} large in the number of summaries being merged.

The merging procedure we propose in this section is closely related to that Berinde et al. However, we address both of the limitations above by giving a much tighter error analysis. Moreover, since our merging
procedure uses our algorithms' Update($i, \Delta$) procedure as a black box, it inherits
the benefits of our Update($i, \Delta$) procedure (e.g., amortized constant runtime) relative to prior algorithms
for weighted streams.

Our procedure has the additional advantage of not requiring the allocation of a new output summary for a merge; rather, when merging two summary via our procedure, we seamlessly update the state of one of the two summaries and output the result (discarding the other summary once merging is complete). 

We mention in passing that other merging procedures for frequent items algorithms have been proposed \cite{merge1, merge2}, but they
provide weaker error guarantees than Berinde et al. (i.e., the error increases with each merge step).
\paragraph*{The Procedure of Agarwal et al. \cite{mergeable}}
Agarwal et al. \cite{mergeable} provided a merging procedure for the \MG\ and \SS\ algorithms that does support arbitrary aggregation trees. In this discussion, we focus on the case of the \MG\ algorithm for simplicity. 
To merge two \MG\ summaries with $k$ counters each, their merging procedure works as follows. The counters from each summary are first added together (so if an item $i$ is assigned a counter in both summaries, its count is set to the sum of the two counters). The counters are then
sorted, and all but the top $k$ counters are discarded. 

Suppose that for $j=\{1, 2\}$, the $j$'th \MG\ summary is a summary of a stream $\sigma_j$ of weighted length $N_j$, and let $N=N_1 + N_2$ denote the weighted length of the concatenated stream $\sigma := \sigma_1 \circ \sigma_2$. Let $f_i$ denote the frequency of $i$ in $\sigma$. Agarwal et al. proved that the resulting summary returns estimates $\hat{f}_i$ for all $i$ such that: \begin{equation} \label{swish} 0 \leq f_i - \hat{f}_i \leq (N-C)/(k+1),\end{equation} where $C$ is the sum of the counters in the merged summary.

\medskip \noindent \textit{Implementing the Merging Procedure.}
\label{sec:variant}
The natural way to implement the merging procedure of Agarwal et al. is the following. 
First, allocate a new hash table capable of storing up to $2k$ counters.
Second, iterate through the counters $(i, c(i))$ in both summaries, adding all $(i, c(i))$ pairs into the new hash table (if key $i$ already resides in the hash table, then add $c(i)$ to its value). This has the effect of summing counters.
Third, sort the $(i, c(i))$ pairs in the hash table by $c(i)$. Fourth, create a new summary of capacity $k$ and insert the top $k$ $(i, c(i))$ pairs into the summary, discarding earlier summaries and the intermediate hash table used for merging.

We observe that an alternative implementation of the third and fourth steps above is to use Quickselect \cite{quickselect}
to identify the $k$th largest counter $c_k$ in the hash table. One then makes an additional pass through the counters, 
identifying all that are at least as large as $c_k$, and feeding each into the new summary. This implementation runs in time $O(k)$.

\medskip \noindent \textit{Disadvantages.} This merging procedure has a number of disadvantages in practice. First, it requires allocating
a hash table of capacity $2k$ to add the counters, which potentially doubles the combined space usage of the two constituent \MG\ summaries. 
Second, the original sorting-based proposal of Agarwal et al. requires
 $\Omega(k \log k)$ time. As we observe above, one can use Quickselect 
 to reduce the runtime to $O(k)$, but the Big-Oh notation hides a substantial constant factor,
 and running Quickselect on up to $2k$ counters 
 is a runtime bottleneck in practice. 
Third, while this merging procedure easily extends to the case where the two summaries have differing numbers of counters (say $k_1$ and $k_2$ counters respectively), it requires $\Omega((k_1 + k_2) \cdot \log(k_1 + k_2))$ time. One can hope to achieve 
$O(\min(k_1, k_2))$ time. 

In this section, we give a different merging procedure that mitigates these disadvantages. Our procedure applies generically to any counter-based algorithm \emph{that can efficiently handle weighted updates}. We now describe the procedure in this level of generality. 

\subsection{Our Merging Procedure}
\paragraph*{Description of our Merging Procedure} For simplicity, let us assume that both
summaries to be merged are configured to store $k$ counters.
Let $k'$ denote the number of (non-zero) counters in the second summary to be merged. In our merging procedure, the $k'$ non-zero counters from the second summary are treated as $k'$ stream updates, and each stream update is fed into the first summary using the summaries Update() procedure. That is, for each item $i$ assigned a counter $c(i)$ in the second summary, then the function Update($i, c(i)$) is called on the first summary. See Algorithm \ref{code:merge} for pseudocode.

{\small
\begin{algorithm}[t]
{\small
\caption{\small{Our Merge Procedure}}\label{code:merge}
\begin{algorithmic}[1]{
\STATE // Notation: $T_1, T_2$ are both counter-based summaries,
\STATE // i.e., $T_2$ is a set of items $i$ and approximate counts $c(i)$, and 
\STATE // $T_1$ comes with an Update$_{T_1}()$ procedure.
\STATE \textbf{Function}: Merge($T_1, T_2$):}
\STATE \hspace{2mm} \textbf{for each } $i \in T_2$:
\STATE \hspace{4mm} Update$_{T_1}(i, c(i))$
\STATE \hspace{2mm} \textbf{return} $T_1$
\end{algorithmic}
}
\end{algorithm}
}

\paragraph*{Space Usage} Clearly, this merging procedure avoids the $2k$ extra words of space (on top of the space to store the two summaries being merged) required by the merging procedure of Agarwal et al. \cite{mergeable}. It uses no more space than that already used by the two summaries being merged, 
 seamlessly updating the state of one of the summaries, outputting the result, and discarding the other summary upon completion. 
 
\paragraph*{Speed} We can guarantee that our merging procedure runs in time $O(k)$. This follows from the fact that
the Update() operation for our algorithms take $O(1)$ time unless a DecrementCounters() operation is triggered, 
and a DecrementCounters() operation takes $O(k)$ time and is triggered at most once every $\Omega(k)$
stream updates. 

One might hope that if $k' \ll k$ (which will be the case, e.g., if the second summary is significantly smaller than the first), then 
the merging procedure runs in $O(k')$ time. Unfortunately, we cannot quite provide this guarantee.
For example, if the second summary contains only a single nonzero counter $c(i)$, and calling Update$(i, c(i))$ on the first summary triggers a DecrementCounters() operation, then the merging procedure will take $\Theta(k)$ time. 
However, we can guarantee the following. If $\Omega(k/k')$ summaries of size at most $k'$ are merged into a single summary of size $k$,
then the amortized time per merge is $O(k')$.

\paragraph*{Error} We now show that our merging procedure, when applied to  our \MED\ algorithm (Algorithm \ref{code:firstweighted}) satisfies an error property analogous to that established by Agarwal et al. (cf. Equation \eqref{swish}). Using similar ideas, it can be shown that our merging procedure satisfies the analogous guarantees  when applied to Algorithm \ref{code:final}, but we omit this result for brevity. While our merging procedure satisfies many desirable properties
compared to that of Agarwal et al. \cite{mergeable}, our error analysis exploits the same key ideas as theirs.
%
\begin{theorem}
Suppose that for $j=\{1, 2\}$, we run Algorithm \ref{code:firstweighted} with $k$ counters on a stream $\sigma_j$ of weighted length $N_j$. Let $N=N_1 + N_2$ denote the weighted length of the concatenated stream $\sigma := \sigma_1 \circ \sigma_2$. Let $f_i$ denote the frequency of $i$ in $\sigma$. Then after applying the merging procedure of Algorithm \ref{code:merge}, the resulting summary can, for any $i$, return an estimate $\hat{f}_i$ satisfying \vspace{-2mm} \begin{equation} \label{swishyours} 0 \leq f_i - \hat{f}_i \leq (N-C)/k^*,\end{equation} where $C$ is the sum of the counter values in the merged summary.
In fact, for any $j < k^*$, it holds that \begin{equation} \label{swish2yours} 0 \leq f_i - \hat{f}_i \leq N^{\text{res}(j)}/k^*.\end{equation}
\end{theorem}
\begin{proof}
Clearly $0 \leq f_i - \hat{f}_i $ for all $i$. We now turn to establishing that $f_i - \hat{f}_i \leq (N-C)/k^*$.

For $j \in \{1, 2\}$, let $\hat{f}_{i, j}$ denote the estimate for $i$ returned by summary $j$ after being run on $\sigma_j$,
and let $f_{i, j}$ denote the true frequency of item $i$ in $\sigma_j$. Let
$E_j := \max_i f_{i, j} - \hat{f}_{i, j}$, and let $c(j)$ denote the sum of the counter values in summary $j$ after being run on $\sigma_j$. Let $E = \max_i f_{i} - \hat{f}_{i}$ denote the maximum error in any estimate returned by 
the merged summary. Finally, let $S$ denote the sum of all values $c_{k^*}$ that
are computed during DecrementCounters() operations called during the merge procedure. 

It is clear that $E \leq E_1 + E_2 + S$. By Lemma \ref{key}, 
\begin{equation} \label{soclose} E \leq (N_1 - C_1)/k^* + (N_2 - C_2)/k^* + S.\end{equation} 
Recall from the proof of Lemma \ref{key} that any DecrementCounters() operation 
decreases at least $k^*$ counters by $c_{k^*}$. Hence, $S \leq (C_1 + C_2 - C)/k^*$.
Combining this with Equation \eqref{soclose}, we conclude that 
$E \leq (N_1 - C_1)/k^* + (N_2 - C_2)/k^* +(C_1 + C_2 - C)/k^* = (N_1 + N_2 - C)/k^* = (N-C)/k^*$,
establishing Equation \eqref{swishyours}.
To see that Equation \eqref{swish2yours} follows from Equation \eqref{swishyours}, use the same argument from the proof of
Theorem \ref{thm:errorweighted} (as the proof of Theorem \ref{thm:errorweighted} showed that the desired tail guarantee follows in a black box manner from Lemma \ref{key}).
\end{proof}
We remark that a final advantage of our merging procedure is that it does not require significant additional
code on top of the streaming algorithm itself: it simply invokes the streaming algorithm's
Update() procedure in a black-box manner.

\medskip
\noindent \textbf{Note.} A subtlety regarding implementations of our proposed merging procedure is as follows. 
If counters are stored in hash tables, then the performance 
guarantees of the 
hash table assume that hash values are (roughly) uniformly distributed throughout the table.
If our 
merge procedure is used 
to merge two summaries whose hash tables use \emph{the same} hash function, 
then care must be taken to ensure that this property continues to hold.
For example, if one iterates through
the second summary's hash table from front to back, calling Update($i, c(i)$) on the first summary
for each counter $(i, c(i))$ encountered, then one may ``overpopulate'' the front of the first summary's hash table.

One avoids this issue entirely if  the summaries choose their hash functions 
independently of each other. Even for summaries that do use the same hash function, 
this issue can be addressed by, say, iterating through the second summary's counters
in a random order. 

\label{sec:merging}
\section{Experiments}
\label{sec:experiments}
\label{sec:expts}
The goal of this section is to investigate three issues. 
First, we compare our main algorithm \SMED\ to alternatives: the proposals \RBMC\ and \MHE\ of prior work, and a variant of $\SMED$ that we call \SMIN. The difference between \SMIN\ and \SMED\ is that the DecrementCounters() operation of \SMIN\ decrements all counters by the sample minimum counter value, instead of the sample median value.
 Notice that \SMIN\ is identical to \RBMC, except that the DecrementCounters() operation of \RBMC\ decrements by the global minimum counter value,
while \SMIN\ decrements by the minimum counter value in a sample of size $\ell=1024$.

We find that \SMED\ is significantly faster than the the algorithms of prior work (5.5x-8.7x faster than \MHE\ and up to 70x faster than \RBMC). While it is not quite as accurate as \MHE\ or \RBMC\ for a given amount of memory, the maximum error
of \SMED\ is at most 30\% larger than that of \MHE, and less than 2.5x that of \RBMC\ and \SMIN\ (the latter two algorithms had nearly identical
error on the datasets used in our experiments).
Moreover, the increased error of \SMED\ relative to alternatives can be overcome, while simultaneously decreasing runtime, by increasing the number of counters by a small factor (less than 2x). This is entirely expected, as our error analysis (cf. Theorems \ref{thm:errorweighted} and \ref{thm:errorfinal}) indicate that for sufficiently large settings of the parameter $\ell$, \SMED\  with $k$ counters has error that is no worse than \RBMC\ with $\approx k/2$ counters. 
Hence, we find that our novel \SMED\ algorithm is the preferable option, except possibly in settings where space and accuracy are paramount. 

Even in settings in which space and accuracy are paramount, 
we find that on the datasets with which we experimented, the preferable solution is our novel algorithm \SMIN. 
For a given space bound, \SMIN\ achieves error nearly identical to
\RBMC\ (and much better than \MHE), with far superior speed. 

Second, our main algorithm \SMED\ (cf. Algorithm \ref{code:final}) naturally allows for a smooth tradeoff between error and speed
as follows. While \SMED\ itself implements the DecrementCounters() procedure by sampling many counters and
decrementing by the sample median, one could instead decrease by a different sample quantile (at the extreme end is \SMIN, which decrements by the sample minimum). A larger decrement amount increases the error, but ensures that DecrementCounters() operations occur less frequently, and hence speeds up the algorithm. 
We investigate this tradeoff on real and synthetic datasets, and find that decrementing by sample medians 
is indeed an attractive point on this tradeoff curve.

Third, we compare our merging procedure to prior work \cite{mergeable}, and find that ours is significantly faster (up to 12x faster than a straightforward implementation of prior work \cite{mergeable}) while using less than half the space.

\subsection{Details}
All experiments were run on a machine running Windows 10 with a Intel Core i5-6600K 3.5GHz Quad-Core Processor, (2 x 8GB) DDR4-2666 Memory, and a 500GB SSD. Test data was generated from the CAIDA Anonymized Internet Traces 2016 Dataset.  
Four packet capture files were chosen at random from the dataset. These were then preprocessed into updates of the form ($a_i, \Delta_i$), where $a_i$ is the source IP with decimal points excluded, and $\Delta_i$ is the packet size in bits. These four preprocessed streams were then concatenated into one larger stream.  The stream length was $n \approx 126.2$ million and the weighted stream length was $N := \sum_{i=1}^{n} \Delta_i \approx 72.2 \times 10^{9}$. Since the data came from an IPv4 packet capture, the universe size is $m=2^{32}$. Although IPv4 addresses can be stored with 32 bits, our implementations, for purposes of generality, use a long long data type (64 bits) to store each. The implementations can be trivially modified to only use 32 bits per identifier if desired. 

The number of unique identifiers in the packet stream was approximately 1.75 million. This may seem small, as the trivial (exact) algorithm that keeps a hash table storing an exact count for each unique IP address in the stream will use only several dozen megabytes of space. Nonetheless, this is more than enough data to obtain a clear qualitative comparison of the tested algorithms. Moreover, our new algorithms, when run with, say, $k=24,576$ counters, use less than 1/70th of the space of the trivial solution.

We also ran experiments on synthetic data generated by a Zipfian distribution with various skewness parameters.
The algorithms performed entirely similarly on these datasets and the packet trace data. Hence, for brevity we
only present our experimental results on the packet trace data.

\subsection{Properties Exhibited By All Algorithms}
\label{sec:commonalities}
For any fixed stream, as the number of counters $k$ used in a counter-based algorithm increases, the algorithm grows closer and closer to
being exact, and the number of DecrementCount() operations relative to the number of stream updates decreases. 
Hence, the algorithms' performances (in both speed and error) converge as the number of counters grows. Consequently, the biggest differences between 
the counter-based algorithms explored in our experiments are observed when the number of counters is relatively small.

\subsection{Comparison to Baselines}

\begin{figure}
\begin{center} 
\includegraphics[width=3.2in]{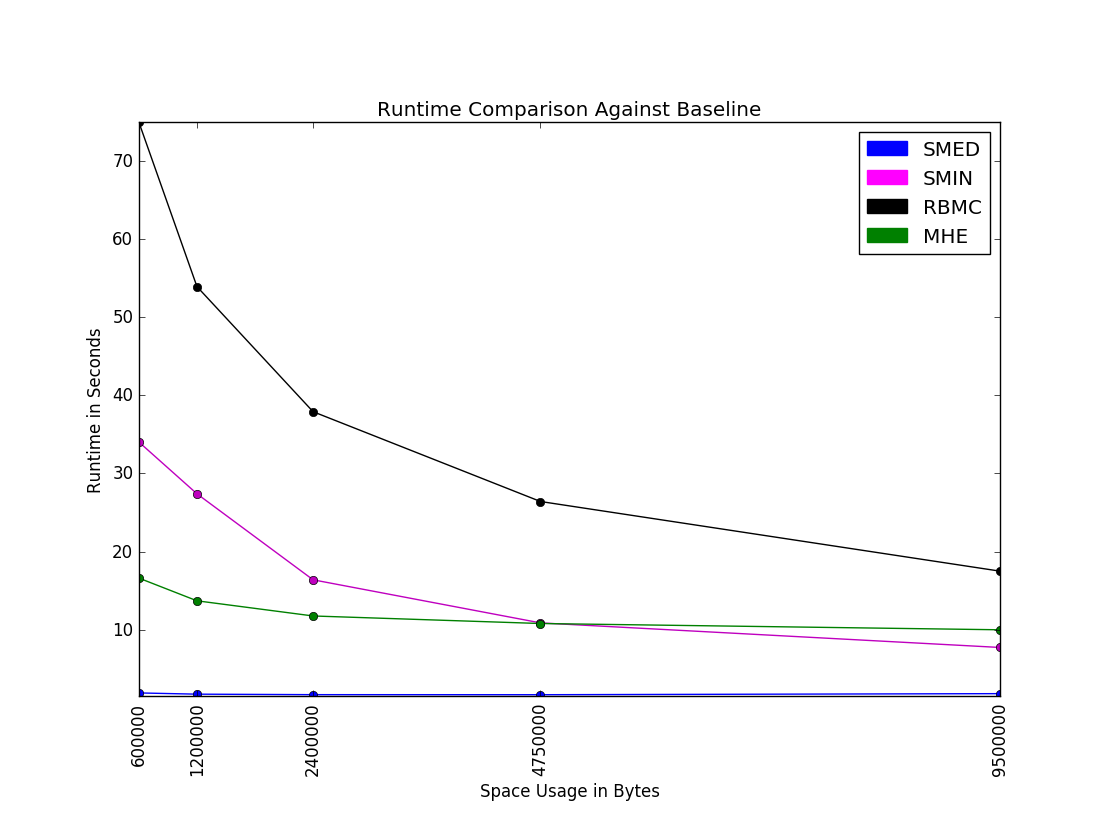}\\
\includegraphics[width=3.2in]{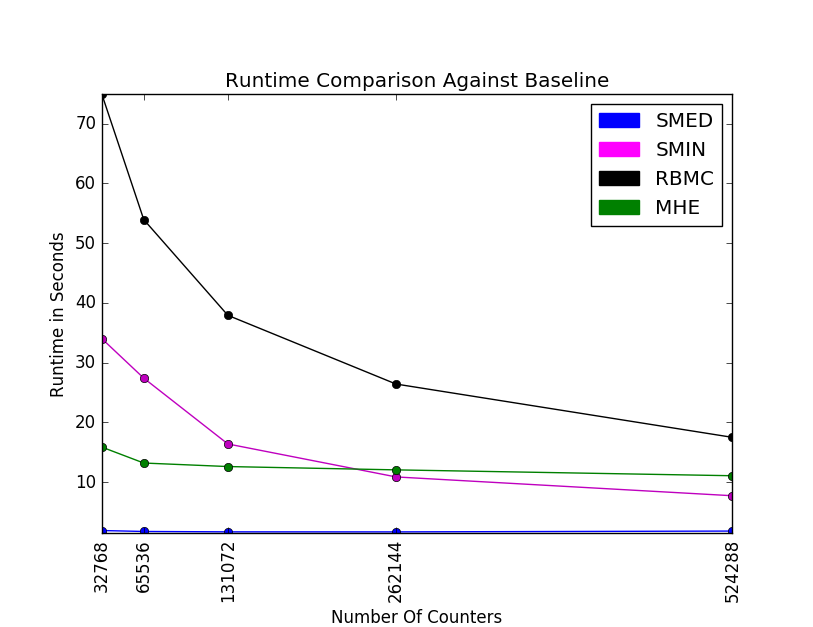}
\end{center} 
\caption{Runtime Comparison of Four Algorithms}
\label{fig:baselinetime}
\end{figure}
\label{sec:baselines}

\begin{figure}
\begin{center} 
\includegraphics[width=3.2in]{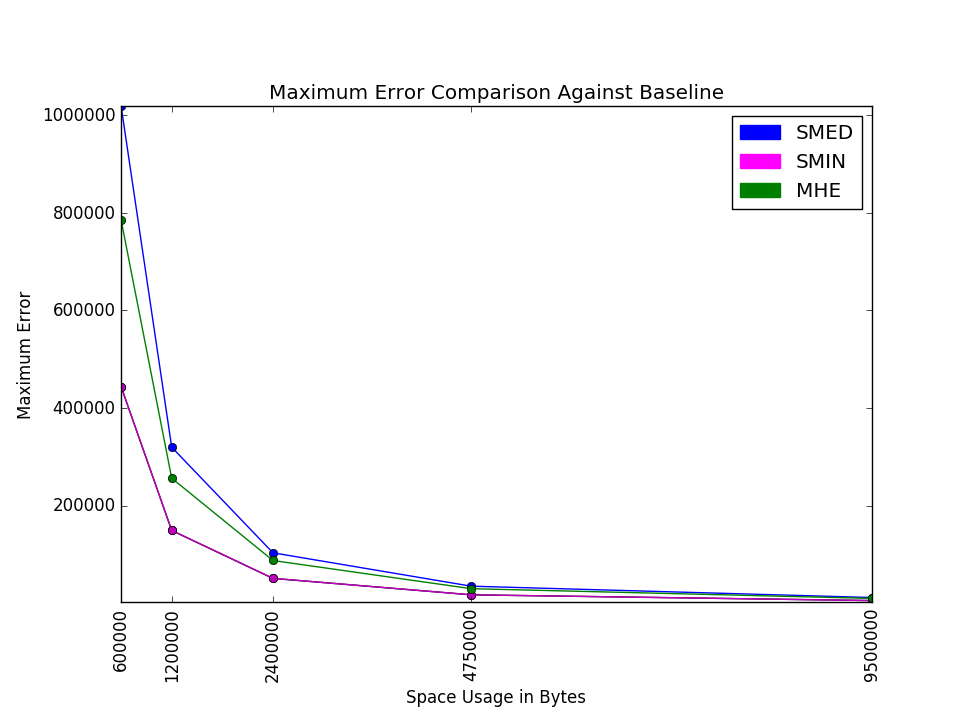}
\mbox{                          }
\includegraphics[width=3.2in]{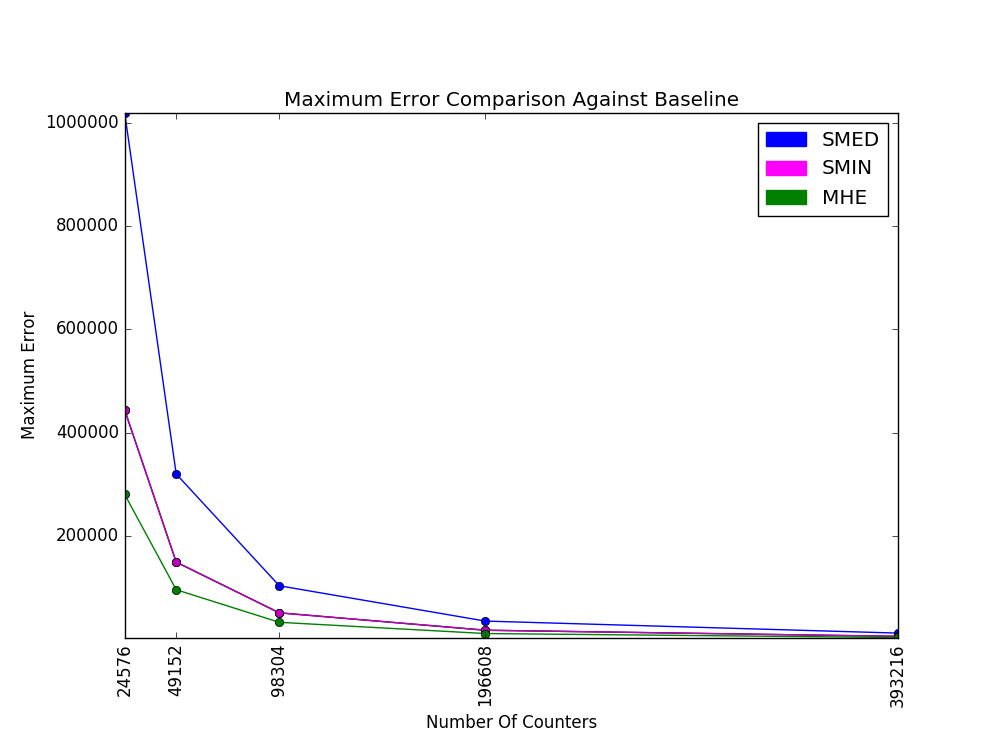}
\end{center} 
\caption{Maximum Error Comparison of Four Algorithms. \RBMC\ is omitted from both plots, as its error
is indistinguishable from \SMIN. For an equal number of counters $k$ (bottom plot), \RBMC, \MHE, and \SMIN\
all have indistinguishable maximum error, and hence only one is displayed. Indeed, \RBMC\ with $k$ counters 
is isomorphic to \MHE\ with $k+1$ counters,
cf. Section \ref{sec:isomorphic}.}
\label{fig:baselineerror}
\end{figure}
We compared our new algorithms, \SMED\ and \SMIN, to the
algorithms for weighted streams given in prior work, \RBMC\ and \MHE. 
While \RBMC, \SMED, and \SMIN\ all use the same amount of space (in bytes)
for a given number of counters $k$, \MHE\ uses additional space owing to the need
to maintain a min-heap data structure in addition to a hash table storing all counters.
 Figures  \ref{fig:baselinetime} and \ref{fig:baselineerror} present both 
equal-space and equal-counters comparisons of the algorithms.

For each of our five tested values of $k$, we ran each algorithm 10 times on the same packet stream, and
all numbers reported are averages of all 10 runs (the error and space usage for the \MHE\ and \RBMC\ algorithms are deterministic functions of the stream, so only the timings varied from run to run for these algorithms). 

Figures  \ref{fig:baselinetime} and \ref{fig:baselineerror} show runtime and maximum error comparisons
for the four algorithms. Only the runtime graphs contain the results for \RBMC, because  
the error of \RBMC\ and \SMIN\ were so close as to be indistinguishable. 

\medskip
\noindent \textbf{Comparison of \SMED\ to Alternatives.}
 Figure \ref{fig:baselinetime} demonstrates that \SMED\ significantly outperforms all three alternatives in terms of runtime. For an equal amount of space, \SMED\ was faster than \MHE\ by a factor of 5.5x-8.7x. Compared to \SMIN, \SMED\ was 6.5x-30x faster, and compared to \RBMC, \SMED\ was 20x-70x times faster. As indicated in Section \ref{sec:commonalities}, the smaller the number of counters $k$, the faster \SMED\ is relative to the alternatives.
 
Figure \ref{fig:baselineerror} demonstrates that for, an equal amount of space, the maximum error of \SMED\ is only 18\%-29\% larger than that of \MHE. \RBMC\ and \SMIN\ both have significantly less error than \SMED\ and \MHE\ for a given amount of space. However, the maximum error of \SMED\ is never more than 2.5x that of \RBMC\ and \SMIN. Moreover, the figure
demonstrates that
this increased error of \SMED\ relative to alternatives can be overcome (while simultaneously decreasing runtime) by increasing the number of counters (and hence space consumption) by a factor of less than 2x.

In conclusion, \SMED\ appears to be the most attractive option, except in settings where space and error are
significantly more important than speed. In these settings, our new \SMIN\ algorithm is more attractive
than the alternatives from prior work, namely \MHE\ and \RBMC\ (we note that Section \ref{sec:tradeoffs} below demonstrates that it is possible
to smoothly interpolate between \SMED\ and \SMIN, obtaining different tradeoffs between speed and error. Exactly
which tradeoff is preferable depends on the relative importance of speed vs. error in the application domain). 
Figures \ref{fig:baselinetime} and \ref{fig:baselineerror} demonstrate that for an equal amount of space,
\SMIN\ is about 2x faster than \RBMC, and their error behaviors are indistinguishable. For an equal amount of space,
\SMIN\ is never more than 2x slower then \MHE\ and is actually faster than \MHE\ for large $k$; meanwhile, the error
of \MHE\ is 1.6x-1.8x that of \SMIN.

\subsection{Tradeoffs Between Speed and Error}
\label{sec:tradeoffs}
Recall that \SMED\ takes a sample of the keys currently being stored in the hash table, uses Quickselect \cite{quickselect} to compute the median of the sample, then decrements all counters by this quantity. We ran tests on a large number of variations of this novel algorithm, where a variation might decrement by other sample quantiles rather than the sample median. Overall, we tested fifty total variations, ranging from the 0th quantile to the 98th quantile. Note that decrementing by the 0th quantile yields the algorithm \SMIN.

The leftmost plot in Figure \ref{fig:quantilesruntime} shows the runtime in seconds for each setting of $k$, and how it varies with the quantile used to determine decrements. The runtimes for \SMIN\ and a few other low quantiles are not displayed so as not to distort the graph, as those runtimes are significantly larger than the rest. This is expected, as 
the less aggressively DecrementCounters() decreases counter values, 
the greater the number of DecrementCounters() operations that occur.
From the 0th quantile (\SMIN) to the 50th quantile (\SMED), runtime drops heavily (specifically, as indicated previously in Section \ref{sec:baselines}, \SMED\ is 6.5x-30x faster than \SMIN, depending on number of counters $k$.) 
Trivially the best runtimes are achieved by using very high quantiles in the DecrementCounters() operation, but  the ``returns'' are diminishing: using the 98th quantile yields an algorithm that is only 20-30\% faster than using the the 20th quantile.

The middle and bottom plot in Figure \ref{fig:quantilesruntime} displays the relationship between quantile and maximum error. As expected, the trend here is that as quantile increases, error increases. However, the error increases relatively slowly as the quantile increases from 0 to about 70, before shooting up rapidly thereafter. 
We conclude that, on the datasets tested, any quantile less than roughly 70 represents a reasonable choice in the tradeoff
curve between speed and error; exactly which point on the curve is most desirable will
depend on the relative importance of speed vs. error in an application.

We believe, however, that it is wise to always use a quantile of roughly 10 or higher
(i.e., we believe this is preferable to using \SMIN). 
Such a choice of quantile achieves almost all of the accuracy advantages of \SMIN, while coming
with significantly improved worst-case runtime bounds compared to $\SMIN$.
In particular, this provides runtime guarantees 
 even for data distributions that 
differ substantially from any that we experimented on.

\begin{figure}
\begin{center} 
\includegraphics[width=3.2in]{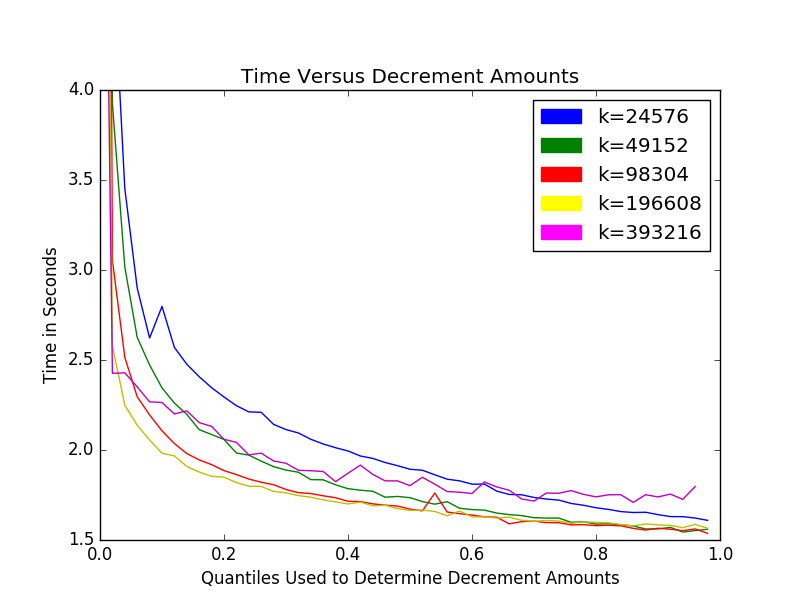}
\includegraphics[width=3.2in]{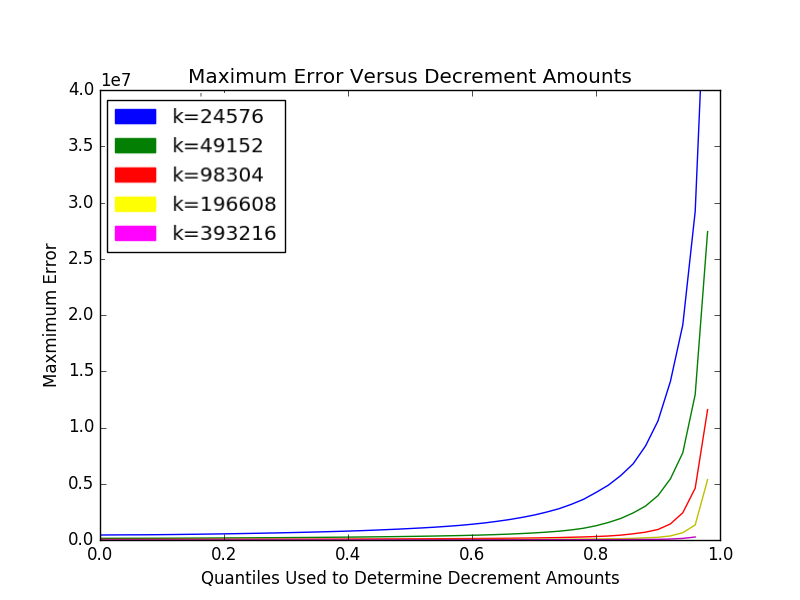}
\includegraphics[width=3.2in]{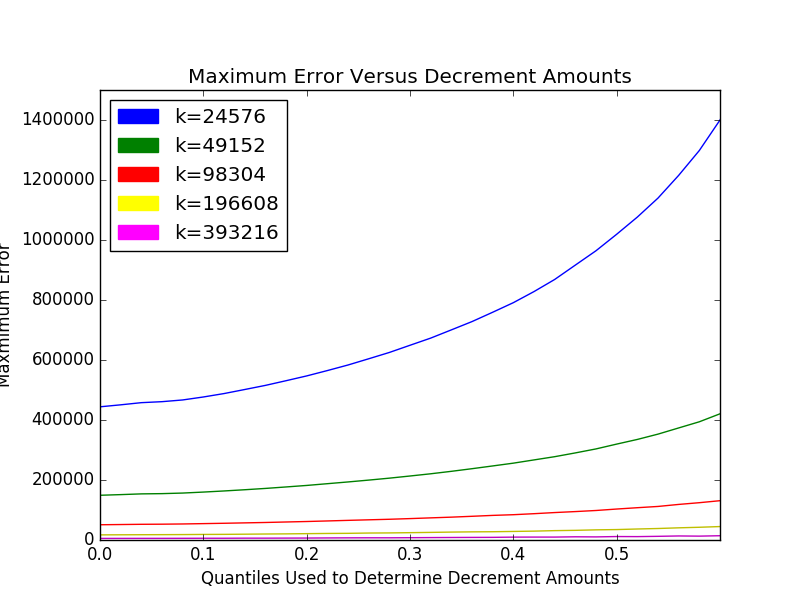}
\end{center} 
\caption{Time and error of our algorithms as a function of the quantile Used by DecrementCounters(). The middle and bottom plots are identical, except that the bottom plot cuts off large quantiles in order to better illustrate 
the error for small quantiles.}
\label{fig:quantilesruntime}
\end{figure}

\begin{figure}
\begin{center} 
\includegraphics[width=3.2in]{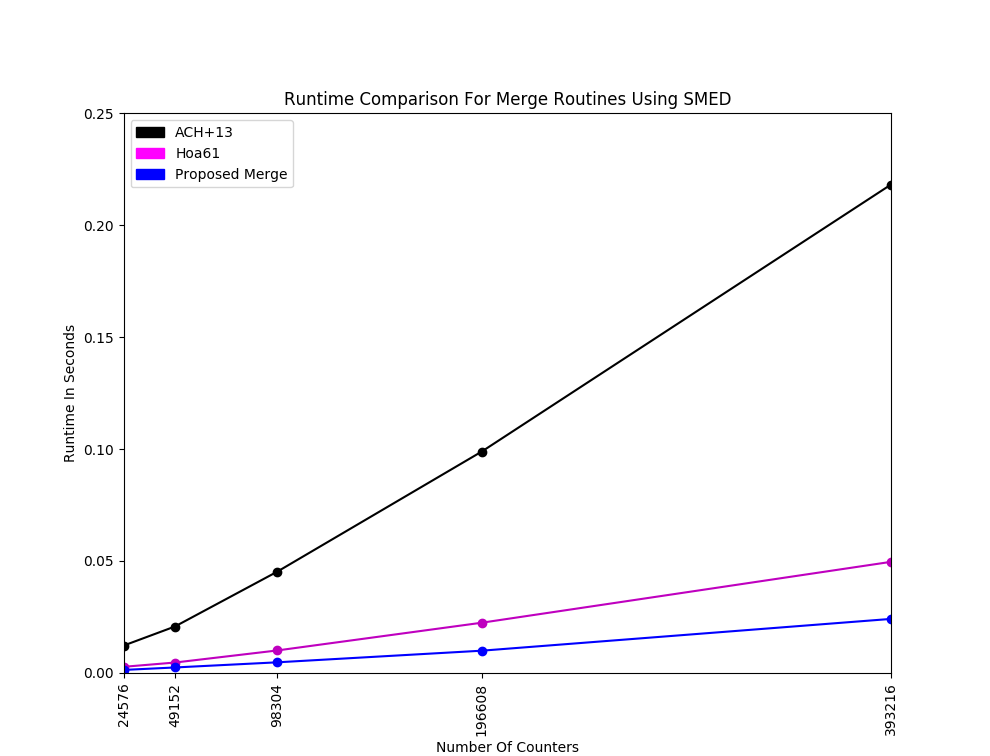}
\end{center} 
\caption{Speed of Our Merge Procedure Compared to Prior Work.
ACH+13 denotes a direct implementation of the merging procedure
of Agarwal et al. \cite{mergeable}. Hoa61 denotes our proposed
Quickselect-based implementation of the procedure of \cite{mergeable} (cf. Section \ref{sec:variant}).}
\label{fig:merge}
\end{figure}

\subsection{Merging}
In addition to the prior tests, we also ran experiments to evaluate our merging procedure (cf. Section \ref{sec:merge}). 
We compared our proposed procedure against the proposal of Agarwal et al. \cite{mergeable}, 
as well as the variant that we proposed (cf. Section \ref{sec:variant}) that uses Quickselect \cite{quickselect} rather than
sorting to identify the $k$th largest counter after ``adding'' the counters from the two sketches being merged.
Recall that our merging procedure uses the Update() operation of a counter-based
algorithm as a black box; we tested our merging procedure using the Update() procedure of \SMED. 

In each of our experiments, we merged 50 pairs of sketches, with each sketch having a capacity of $k$ counters. 
These experiments require 50x the amount of raw data that is required to test a single sketch. Hence, we ``filled up'' the sketches using synthetic data streams before merging them: the streams
had item identifiers drawn from a Zipfian distribution with parameter $\alpha=1.05$, and item wieghts were generated from a uniformly random distribution from 1 to 10,000 (cf. \cite[Section 5]{berinde} for a definition of this distribution).

As illustrated in Figure \ref{fig:merge}, when using \SMED, our proposed merging procedure runs up to 8.6x-10x faster than the procedure of \cite{mergeable} (the bigger the sketch, the faster our proposed procedure is relative to that of \cite{mergeable}). 
Our
proposed merging procedure is 1.9x-2.26 faster than the Quickselect-based variant. 
Moreover, the error of our merge procedure when using \SMED\ was always within 2.5\% of the error of the 
procedure of \cite{mergeable} and the equivalent Quickselect-based variant; owing to
the small magnitude of this difference, we do not display this error graphically.
 
It is worth mentioning that our proposed merge algorithm also uses less space than the two alternatives.  Both alternatives require allocating an additional hash table of capacity $2k$ in order to store all counters from both sketches, as well as an extra hash table of capacity $k$ to store the final merged sketch. Hence, they consume 2.5x more space than our procedure, which uses no additional space beyond that already used by the two sketches to be merged.

\section{Other Prior Work}
\paragraph{Work of Bhattacharyya et al.}
Recent work of Bhattacharyya et al. \cite{woodruffpods} gives improved space bounds for
identifying $(\phi, \eps)$-heavy hitters in the $\ell_1$ norm for insertion streams. 
They describe their algorithms in the context of unweighted streams.
Roughly speaking, for any constant failure probability $\delta>0$, their algorithms manage to replace the $\log n$ factor in the $O(\eps^{-1} \cdot \log n)$ space bound (in bits)
for the basic \MG\ and \SS\ algorithms with a factor of $\log(\phi^{-1})$, which is asymptotically optimal.\footnote{The precise space bound they achieve for constant failure probability $\delta > 0$ is, in our notation, $O(\eps^{-1} \cdot \log(\phi^{-1}) + \phi^{-1} \cdot \log m + \log\log n)$ bits.}

They give two algorithms for identifying $(\phi, \eps)$-heavy hitters: a relatively simple algorithm
that is nearly space-optimal (it is sub-optimal only because it replaces the $\log n$ factor in the space bound of \MG\ and \SS\
with a factor of $\log(\eps^{-1})$ instead of $\log(\phi^{-1})$), 
and a 
significantly more
complicated algorithm that achieves asymptotically optimal space usage.

Their simple algorithm works by sampling $O(\eps^{-2}\log(1/\delta))$ stream updates at random,
and feeding the sampled updates into a (slightly modified) instance of the \MG\ algorithm with $k=O(1/\eps)$ counters.
The samples are computed implicitly in $O(1)$ time per stream update by
``skipping'' each stream update with probability $1-p$, for a parameter $p=O(\eps^{-2}\log(1/\delta)/n)$. 
To save space, the counters
in the \MG\ 
algorithm do not store full item identifiers. Instead, they each store a hash (consisting of $O(\log 1/\eps)$ bits) 
of the assigned identifier; a separate table is used to track the actual identifiers
corresponding to the $\phi^{-1}$ largest counters in the \MG\ summary. For such identifiers, estimated frequencies
in the sample (scaled appropriately) are used as estimates for their true frequencies.

Some thought is required in order to extend this idea to the setting of weighted streams while maintaining
constant-time updates. Trivially, one can treat an update $(i, \Delta)$ as $\Delta$ updates of the form $(i, 1)$,
but this does not maintain the time bound (even in an amortized sense) if $\Delta = \omega(1)$ on average. 
We now sketch an adaptation that has the same accuracy guarantees
as the simple algorithm of Bhattacharyya et al. \cite{woodruffpods},
and runs in constant amortized time per stream update with high probability,
assuming that the number of stream updates $n$ is $\Omega(\eps^{-2}\log(1/\delta))$.

Assume the weighted stream length $N=\sum_i f_i$
is known in advance (this assumption can be removed using observations of \cite[Section 3.5]{woodruffpods}).
Fix a suitable $p=O(\eps^{-2}\log(1/\delta)/N)$.
 When processing update $(i, \Delta)$, the streaming algorithm repeatedly
 samples geometric random variables with parameter $p$
 (i.e., the variables are distributed according to the number of Bernoulli trials with success probability $p$ 
 that are needed to get one successful trial).
 It stops when the sum of the variables is more than $\Delta$. 
 If $t$ samples are required before the sum of the variables is more than $\Delta$,
 then the algorithm feeds 
update $(i, t)$ into an instance of any counter-based algorithm
capable handling weighted updates. The remainder of the details, and the error
analysis, are similar to \cite{woodruffpods},
and we omit them for brevity.

The above adaptation allows our efficiency improvements relative to prior work for weighted
streams and mergeability to carry over in a black-box
manner to the simple algorithm of Bhattacharyya et al. \cite{woodruffpods}. 
One simply feeds the (weighted) sampled stream updates into our new optimized algorithms for weighted updates.

\paragraph{Work of Bhattacharyya et al.}
Recent independent work of Sivaraman et al. \cite{muthunew} proposes modifying the \SS\ algorithm (cf. Section \ref{sec:ssunit}) as follows. When processing a stream update $(i, \Delta)$, if $i$ is not assigned a counter and all counters are in use, sample $\ell$ counters at random (where $\ell$ is a design parameter), and then reassign this counter to $i$ and increment it by $\Delta$. The primary motivation for this proposal is to minimize memory accesses on each stream update, as memory accesses are a bottleneck in network switching hardware. For $\ell=O(1)$, this proposal runs in constant time per stream update, but may have larger error than our proposals. We leave a detailed experimental comparison of our algorithms with the proposal of Sivaraman et al. \cite{muthunew} to future work.

\label{sec:additional}
\section{Conclusion}
We described a highly optimized version of Misra and Gries' seminal algorithm for estimating frequencies of items over data streams. 
Our algorithm improves on two theoretical and practical aspects of prior work.
First, it handles {\it weighted} updates in amortized constant time.
Second, it uses a simple and fast method for merging summaries that asymptotically improves on prior work even for unweighted streams. 

One direction for future work is to incorporate our optimized algorithms
into more complicated streaming algorithms that use heavy hitter algorithms as subroutines. 
These include estimating the empirical entropy of a data stream \cite{entropy},
and identifying hierarchical heavy hitters \cite{hhh}. 

\vspace{2mm}
\medskip
\noindent \textbf{Acknowledgements.} We are grateful to Graham Cormode for helpful comments on an earlier version of this manuscript.

\bibliographystyle{abbrv}
\bibliography{bibs}  

\begin{thebibliography}{10}

\bibitem{mergeable}
P.~K. Agarwal, G.~Cormode, Z.~Huang, J.~M. Phillips, Z.~Wei, and K.~Yi.
\newblock Mergeable summaries.
\newblock {\em {ACM} Trans. Database Syst.}, 38(4):26, 2013.

\bibitem{berinde}
R.~Berinde, P.~Indyk, G.~Cormode, and M.~J. Strauss.
\newblock Space-optimal heavy hitters with strong error bounds.
\newblock {\em {ACM} Trans. Database Syst.}, 35(4):26, 2010.

\bibitem{woodruffpods}
A.~Bhattacharyya, P.~Dey, and D.~P. Woodruff.
\newblock An optimal algorithm for $\ell_1$-heavy hitters in insertion streams
  and related problems.
\newblock In {\em Proceedings of \emph{PODS}}, pages 385--400, 2016.

\bibitem{beatingcount}
V.~Braverman, S.~R. Chestnut, N.~Ivkin, and D.~P. Woodruff.
\newblock Beating countsketch for heavy hitters in insertion streams.
\newblock In {\em Proceedings of \emph{STOC}}, pages 740--753, 2016.

\bibitem{entropy}
A.~Chakrabarti, G.~Cormode, and A.~McGregor.
\newblock A near-optimal algorithm for estimating the entropy of a stream.
\newblock {\em {ACM} Trans. Algorithms}, 6(3), 2010.

\bibitem{countsketch}
M.~Charikar, K.~Chen, and M.~Farach-Colton.
\newblock Finding frequent items in data streams.
\newblock In {\em Proceedings of \emph{ICALP}}, pages 693--703, 2002.

\bibitem{ch}
G.~Cormode and M.~Hadjieleftheriou.
\newblock Methods for finding frequent items in data streams.
\newblock {\em {VLDB} J.}, 19(1):3--20, 2010.

\bibitem{hhhapp2}
G.~Cormode, F.~Korn, S.~Muthukrishnan, and D.~Srivastava.
\newblock Diamond in the rough: Finding hierarchical heavy hitters in
  multi-dimensional data.
\newblock In G.~Weikum, A.~C. K{\"{o}}nig, and S.~De{\ss}loch, editors, {\em
  Proceedings of the {ACM} {SIGMOD} International Conference on Management of
  Data, Paris, France, June 13-18, 2004}, pages 155--166. {ACM}, 2004.

\bibitem{countmin}
G.~Cormode and S.~Muthukrishnan.
\newblock An improved data stream summary: The count-min sketch and its
  applications.
\newblock In {\em Proceedings of \emph{LATIN}}, pages 29--38, 2004.

\bibitem{entropyapp2}
Y.~Gu, A.~McCallum, and D.~F. Towsley.
\newblock Detecting anomalies in network traffic using maximum entropy
  estimation.
\newblock In {\em Proceedings of the 5th Internet Measurement Conference, {IMC}
  2005, Berkeley, California, USA, October 19-21, 2005}, pages 345--350.
  {USENIX} Association, 2005.

\bibitem{quickselect}
C.~A.~R. Hoare.
\newblock Algorithm 65: Find.
\newblock {\em Commun. ACM}, 4(7):321--322, July 1961.

\bibitem{edo}
E.~Liberty.
\newblock Simple and deterministic matrix sketching.
\newblock In {\em Proceedings of \emph{KDD}}, pages 581--588, 2013.

\bibitem{merge1}
A.~Manjhi, S.~Nath, and P.~B. Gibbons.
\newblock Tributaries and deltas: Efficient and robust aggregation in sensor
  network streams.
\newblock In {\em Proceedings of \emph{SIGMOD}}, pages 287--298, 2005.

\bibitem{merge2}
A.~Manjhi, V.~Shkapenyuk, K.~Dhamdhere, and C.~Olston.
\newblock Finding (recently) frequent items in distributed data streams.
\newblock In {\em Proceedings of \emph{ICDE}}, pages 767--778, 2005.

\bibitem{lossycounting}
G.~S. Manku and R.~Motwani.
\newblock Approximate frequency counts over data streams.
\newblock In {\em Proceedings of \emph{VLDB}}, pages 346--357. VLDB Endowment,
  2002.

\bibitem{spacesaving}
A.~Metwally, D.~Agrawal, and A.~{El Abbadi}.
\newblock Efficient computation of frequent and top-k elements in data streams.
\newblock In {\em Proceedings of \emph{ICDT}}, pages 398--412, 2005.

\bibitem{mg}
J.~Misra and D.~Gries.
\newblock Finding repeated elements.
\newblock {\em Science of Computer Programming}, 2(2):143 -- 152, 1982.

\bibitem{hhh}
M.~Mitzenmacher, T.~Steinke, and J.~Thaler.
\newblock Hierarchical heavy hitters with the space saving algorithm.
\newblock In {\em Proceedings of \emph{ALENEX}}, pages 160--174, 2012.

\bibitem{mu}
M.~Mitzenmacher and E.~Upfal.
\newblock {\em Probability and computing - randomized algorithms and
  probabilistic analysis}.
\newblock Cambridge University Press, 2005.

\bibitem{datasketches}
L.~Rhodes, K.~Lang, A.~Saydakov, J.~Thaler, E.~Liberty, and J.~Malkin.
\newblock Data{S}ketches: {A} {J}ava software library of stochastic streaming
  algorithms, 2017.
\newblock \url{https://datasketches.github.io}.

\bibitem{muthunew}
V.~Sivaraman, S.~Narayana, O.~Rottenstreich, S.~Muthukrishnan, and J.~Rexford.
\newblock Smoking out the heavy-hitter flows with hashpipe.
\newblock {\em CoRR}, abs/1611.04825, 2016.
\newblock To appear in \emph{SDN} 2017.

\bibitem{entropyapp1}
A.~Wagner and B.~Plattner.
\newblock Entropy based worm and anomaly detection in fast {IP} networks.
\newblock In {\em 14th {IEEE} International Workshops on Enabling Technologies
  {(WETICE} 2005), 13-15 June 2005, Link{\"{o}}ping, Sweden}, pages 172--177.
  {IEEE} Computer Society, 2005.

\bibitem{entropyapp3}
K.~Xu, Z.~Zhang, and S.~Bhattacharyya.
\newblock Profiling internet backbone traffic: behavior models and
  applications.
\newblock In R.~Gu{\'{e}}rin, R.~Govindan, and G.~Minshall, editors, {\em
  Proceedings of the {ACM} {SIGCOMM} 2005 Conference on Applications,
  Technologies, Architectures, and Protocols for Computer Communications,
  Philadelphia, Pennsylvania, USA, August 22-26, 2005}, pages 169--180. {ACM},
  2005.

\bibitem{hhhapp1}
Y.~Zhang, S.~Singh, S.~Sen, N.~G. Duffield, and C.~Lund.
\newblock Online identification of hierarchical heavy hitters: algorithms,
  evaluation, and applications.
\newblock In A.~Lombardo and J.~F. Kurose, editors, {\em Proceedings of the 4th
  {ACM} {SIGCOMM} Internet Measurement Conference, {IMC} 2004, Taormina,
  Sicily, Italy, October 25-27, 2004}, pages 101--114. {ACM}, 2004.

\end{thebibliography}

\end{document}